\documentclass{article} [a4paper,11pt]                   
\usepackage{vmargin}
\setmarginsrb{1.1in}{1.1in}{1.1in}{1.1in}{0pt}{5mm}{0pt}{6mm}
\usepackage{graphicx}
\usepackage{amssymb,amsmath,textcomp}
\usepackage{amsthm}
\usepackage{color}
\usepackage{hyperref}
\usepackage{tabularx}
\usepackage{multirow}
\usepackage[justification=centering]{caption}
\usepackage{latexsym}
\usepackage{authblk}
\usepackage{enumerate}
\usepackage[utf8]{inputenc}
\usepackage{lineno}
\usepackage{pifont}
\usepackage{thm-restate}

\newcommand{\remove}[1]{}

\newcommand{\defparproblem}[4]{
 \vspace{3mm}
\noindent\fbox{
  \begin{minipage}{.95\textwidth}
  \begin{tabular*}{\textwidth}{@{\extracolsep{\fill}}lr} \textsc{#1} \\ \end{tabular*}
  {\bf{Input:}} #2  \\
  {\bf{Parameter:}} #3\\
  {\bf{Goal:}} #4
  \end{minipage}
  }
  \vspace{2mm}
}

\usepackage{microtype}

\newtheorem{definition}{\bf Definition}[section]

\newtheorem{reduction rule}{\bf Reduction Rule}[section]
\newtheorem{proposition}{\bf Proposition}[section]
\newtheorem{lemma}{\bf Lemma}[section]

\newtheorem{corollary}{\bf Corollary}
\newtheorem{open problem}{\bf Open Problem}

\newcommand{\mypara}[1]{\smallskip\noindent{\textbf{\sffamily #1} \ }}

\newcommand{\tw}{{\normalfont \textsf{tw}}}

\newcommand{\wb}{{\normalfont \textsf{wb}}}
\newcommand{\cw}{{\normalfont \textsf{cw}}}

\newcommand{\bbF}{{\mathbb{F}}}

\newcommand{\cO}{{\cal O}}

\newcommand{\cI}{\mathcal{I}}

\newcommand{\cH}{\mathcal{H}}

\newcommand{\cX}{\mathcal{X}}
\newcommand{\cM}{\mathcal{M}}

\newcommand{\AAA}{{\mathcal A}}
\newcommand{\GG}{{\mathcal G}}

\newcommand{\SSS}{{\mathcal S}}
\newcommand{\TT}{{\mathcal T}}

\newcommand{\XX}{{\mathcal X}}

\renewcommand{\widetilde}[1]{#1}

\newcommand{\shortversion}[1]{}

\newcommand{\PECVC}{{\sc $p$-Edge-Connected Vertex Cover}}

\newcommand{\PECBTDS}{{\sc $p$-Edge-Connected $\eta$-Treedepth Deletion Set}}

\newcommand{\PEPWOneVd}{{\sc $p$-Edge Connected Pathwidth-1 Vertex Deletion Set}}
\newcommand{\PEBDDS}{{\sc $p$-Edge-Connected $\eta$-Degree Deletion Set}}

\newcommand{\FF}{{\mathcal F}}
\newcommand{\OO}{{\mathcal O}}
\newcommand{\bigoh}{{\mathcal O}}

\newcommand{\nn}{{\mathbb N}}

\newcommand{\ExtPECon}{{\sc Steiner Subgraph Extension}}
\newcommand{\ExtPCon}{{\sc Steiner Subgraph Extension}}

\newcommand{\pEPWOneDS}{{\sc $p$-Edge-Connected One-Pathwidth Deletion Set}}

\newcommand{\pEBoundedVC}{{\sc $p$-Edge-Connected $\eta$-Path Vertex Cover}}

\newcommand{\bdds}{{\sc Bounded Degree Deletion Set}}
\newcommand{\pwoneds}{{\sc One-Pathwidth Deletion Set}}
\newcommand{\tdds}{{\sc $\eta$-Treedepth Deletion Set}}
\newcommand{\pvc}{{\sc $\eta$-Path Vertex Cover}}
\newcommand{\GoneGtwoDeletion}{{\sc $(\GG_1, \GG_2)$-Deletion Set}}

\newcommand{\pbdds}{{\sc $p$-Edge-Con-$\eta$-Deg-DS}}
\newcommand{\ppwoneds}{{\sc $p$-Edge-Con-1PWDS}}
\newcommand{\ptdds}{{\sc $p$-Edge-Con-$\eta$-TDDS}}

\newcommand{\ppvc}{{\sc $p$-Edge-Con-$\eta$-PVC}}
\newcommand{\pgonegtwodel}{{\sc $p$-Edge-Con-$(\GG_1, \GG_2)$-Deletion Set}}

\newcommand{\td}{{\sf td}}

\newcommand{\sv}[1]{}

\usepackage{tikz}
\usetikzlibrary{decorations.shapes,decorations.pathreplacing,decorations.pathmorphing}
\usetikzlibrary{arrows,matrix,shapes}
\usetikzlibrary{positioning}
\tikzset{
        stars/.style={star,inner sep=2pt}
    }

\usepackage{todonotes}
\usepackage[ruled,vlined,linesnumbered]{algorithm2e}

\title{Highly Connected Steiner Subgraph -- Parameterized Algorithms and Applications to Hitting Set Problems\footnote{A preliminary version \cite{EibenMR23} of this paper appeared in proceedings of 48th International Symposium on Mathematical Foundations of Computer Science (MFCS) 2023. Research of Diptapriyo Majumdar has been supported by Science and Engineering Research Board (SERB) grant SRG/2023/001592. Research of M. S. Ramanujan has been partially supported by Engineering and Physical Research Council (EPSRC) grants EP/V007793/1 and EP/V044621/1.}}

\author[1]{Eduard Eiben}
\author[2]{Diptapriyo Majumdar}
\author[3]{M. S. Ramanujan}
\affil[1]{Royal Holloway, University of London, Egham, United Kingdom\\
  \texttt{eduard.eiben@rhul.ac.uk}}
\affil[2]{Indraprastha Institute of Information Technology Delhi, New Delhi, India\\
	\texttt{diptapriyo@iiitd.ac.in}}
\affil[3]{University of Warwick, Coventry, United Kingdom\\
	\texttt{R.Maadapuzhi-Sridharan@warwick.ac.uk}}

\providecommand{\keyword}[1]{\textbf{{Keywords:}} #1}

\begin{document}
\maketitle

\begin{abstract}
Given a simple (connected) undirected graph $G$, a set $X \subseteq V(G)$ and integers $k$ and $p$, the {\ExtPECon} problem asks whether there exists a set $S \supseteq X$ of at most $k$ vertices such that $G[S]$ is a $p$-edge-connected subgraph. This problem is a natural generalization of the well-studied {\sc Steiner Tree} problem (set $p=1$ and $X$ to be the terminals). In this paper, we initiate the study of {\ExtPECon} from the perspective of parameterized complexity and give a fixed-parameter algorithm (i.e., FPT algorithm) parameterized by $k$ and $p$ on graphs of bounded degeneracy.
In case we remove the assumption of the graph being bounded degenerate, then {\ExtPCon} becomes W[1]-hard.
Besides being an independent advance on the parameterized complexity of network design problems, our result has natural applications. In particular, we use our result to obtain new single-exponential FPT algorithms for several vertex-deletion problems studied in the literature, where the goal is to delete a smallest set of vertices such that:
\begin{enumerate}[(i)]
	\item the resulting graph belongs to a specified hereditary graph class, and
	\item the deleted set of vertices induces a $p$-edge-connected subgraph of the input graph. 
\end{enumerate}
\end{abstract}

\medskip

\keyword{Parameterized Complexity, Steiner Subgraph Extension, Matroids, Representative families, $p$-edge-connectivity}

\section{Introduction}
\label{sec:intro}

Given a simple undirected graph $G = (V, E)$ and a set $T \subseteq V(G)$, called {\em terminals}, the {\sc Steiner Tree} problem asks if there are at most $k$ edges $F \subseteq E(G)$ such that there is a path between every pair of vertices of $T$ in $G' = (V, F)$.
{\sc Steiner Tree} is one of the fundamental problems in network design and is a well-studied problem in parameterized complexity (\cite{FeldmannM16,FigueiredoLMS22,ChitnisFHM20,BateniDHM16,MarxPP18}).
We refer to Section \ref{sec:prelims} for the definitions related to parameterized complexity and graph theory.
In this paper, we study the {\ExtPECon} problem, which we formally define below.

\defparproblem{{\ExtPECon}}{A simple undirected graph $G = (V, E)$, $X \subseteq V(G)$ and integers $k,p\in {\mathbb N}$.}{$k + p$}{Is there $S \supseteq X$ of size at most $k$ such that $G[S]$ is $p$-edge-connected?}

{Note that {\sc Steiner Tree} puts a restriction that a set $F$ of at most $k$ edges can be used to make a set of terminals $X$ connected.
This implicitly puts a restriction that the total number of endpoints of $F$ is at most $k+1$.
It is equivalent to asking if there exists a set $S \supseteq X$ of at most $k+1$ vertices such that $G[S]$ is connected.
Hence, {\ExtPCon} generalizes the notion of {\sc Steiner Tree}.}
To the best of our knowledge, the parameterized complexity status of {\ExtPCon} is unexplored even for arbitrary fixed constant $p$.
Our problem {\ExtPCon} is closely related to a special variant of {\sc Edge-Connected Survivable Network Design} (defined by Feldmann et al. \cite{Feldmann0L22}) problem.
The goal of that variant of the {\sc Edge-Connected Survivable Network Design} is to find a collection of `at most $k$ edges' so that there are $p$ edge-disjoint paths between every pair of vertices in the terminal set.
Abhinav et al. \cite{AbhinavBBS23} studied the {\sc Steiner Subgraph Extension} problem when $p = n - k$ and $k$ is the parameter.
It means that they aim to find an $(n-k)$-edge-connected Steiner subgraph with exactly $\ell$ vertices.
In our problem, observe that $p \leq k - 1$ as any simple graph with $k$ vertices cannot be $p$-edge-connected for any $p \geq k$.
If we set $p = k-1$, then our problem  becomes precisely the same as {\sc Clique} problem, where we want to decide if a graph has a clique with exactly $k$ vertices, a W[1]-hard problem.
Hence, one must place further restrictions on the input when aiming for fixed-parameter tractability. 
In this paper, we consider the {\ExtPCon} problem when $\eta$ is the degeneracy of the input graph $G$ and $\eta$ is a fixed-constant.
Note that many well-known sparse graph classes are subclasses of bounded degenerate graphs.
For instance, planar graphs are $5$-degenerate, graphs with treewidth (or treedepth or pathwidth) at most $\eta$ are $(\eta+1)$-degenerate.

\medskip

\mypara{Our Contributions:}
The input to {\ExtPCon} is an $\eta$-degenerate simple undirected graph with $n$ vertices such that $\eta$ is a fixed constant.
Recall that the parameter is $k + p$.
The first part of our paper proves that {\ExtPCon} is fixed-parameter tractable (FPT) when the input graph has constant degeneracy.
In particular, we give an FPT algorithm with running time $2^{\OO(pk + \eta)}n^{\OO(1)}$-time for {\ExtPECon} when the input graph $G$ is $\eta$-degenerate. The formal statement of the theorem is given below.

\begin{restatable}{theorem}{ThmMain}
\label{thm:main-result}
{\ExtPECon} can be solved in time $2^{\OO(pk + \eta)}n^{\OO(1)}$, where $\eta$ is the degeneracy of the input graph.
\end{restatable}

It follows that for graphs of constant degeneracy and for constant $p$, the above result gives a $2^{\bigoh(k)}n^{\OO(1)}$-time algorithm, which is useful in several applications as we show in this paper.  
The above result  crucially relies on the use of the {\em out-partition} matroid, its linear representability in deterministic polynomial-time, and a dynamic programming subroutine using the notion of representative sets.
We would like to highlight that Einarson et al. \cite{EinarsonGJMW23} have studied the same problem when $X$ is a vertex cover.
Our dynamic programming algorithm over representative sets has some similarities with the algorithm of Einarson et al. \cite{EinarsonGJMW23} but $X$ is not necessarily a vertex cover of $G$ for {\ExtPCon}.
For our problem, $X$ is the set of {\em terminals}.
Despite the fact that $G$ is a bounded degenerate graph, designing an algorithm for {\ExtPCon} needs careful adjustment to the subproblem definitions and some additional conditions have to be incorporated while constructing the collection of sets in the DP formulation.
Furthermore, our algorithm in Theorem \ref{thm:main-result} does not depend on $\eta$ in the exponent of $n$.
As a straight-forward consequence of our main result Theorem \ref{thm:main-result}, we have the following corollary.

\begin{restatable}{corollary}{CorMain}
\label{corollary:main-result}
{\ExtPCon} admits an algorithm that runs in $2^{\OO(pk + \eta)}n^{\OO(1)}$-time when one of the following conditions is satisfied. 
\begin{enumerate}[(i)]
	\item $\eta$ is the treewidth of the input graph.
	\item $\eta$ is the cutwidth of the input graph when a linear arrangement of optimal width is given as part of the input.
\end{enumerate}	
\end{restatable}

The second part of our paper describes some applications of our main result (Theorem \ref{thm:main-result}) to some natural problems in graph theory with connectivity constraints.
Einarson et al. \cite{EinarsonGJMW23} have initiated the study of {\PECVC} that asks to find a vertex cover $S$ of $G$ containing at most $k$ vertices such that $G[S]$ is $p$-edge-connected.
In their paper, they proved that {\PECVC} can be solved in $2^{\OO(pk)}n^{\OO(1)}$-time.
Being motivated by their result, we illustrate how Theorem \ref{thm:main-result} lays us a foundation to {\em design efficient deterministic parameterized singly exponential-time algorithms} for {\bdds}, {\tdds}, {\pwoneds} and {\pvc} with $p$-edge-connectivity constraints.
Each of these problems are well-studied without the connectivity constraints (see \cite{GiannopoulouJLS17,CervenyS19,CervenyS21} for more details) in parameterized and approximation algorithms.
We state the problem definitions below.
Given an undirected graph $G = (V, E)$, the following questions are asked by these problems.
\begin{itemize}
	\item {{\PEBDDS} ({\pbdds}) asks if there is $S \subseteq V(G)$ with at most $k$ vertices such that $G - S$ is a graph of maximum degree at most $\eta$ and $G[S]$ is $p$-edge-connected}.
	\item {\PECBTDS} ({\ptdds}) asks if there is $S \subseteq V(G)$ with at most $k$ vertices such that $G - S$ has treedepth at most $\eta$ and $G[S]$ is $p$-edge-connected.
	\item {\pEPWOneDS} ({\ppwoneds}) asks if there is $S \subseteq V(G)$ with at most $k$ vertices such that $G - S$ has pathwidth at most 1 and $G[S]$ is $p$-edge-connected.
	\item {\pEBoundedVC} ({\ppvc}) asks if there is $S \subseteq V(G)$ with at most $k$ vertices such that $G - S$ has no $P_{\eta}$ as subgraph and $G[S]$ is $p$-edge-connected.
\end{itemize}

\medskip

Jacob et al. \cite{JacobKMR23,JacobMR23} have initiated the study of vertex deletion to scattered graph classes.
Given an undirected graph $G = (V, E)$ and two hereditary graph classes $\GG_1, \GG_2$, the {\GoneGtwoDeletion} problem asks if there is a vertex subset $S \subseteq V(G)$ with at most $k$ vertices such that every connected component of $G - S$ is either in $\GG_1$ or in $\GG_2$.
Jacob et al. \cite{JacobMR23} have proved several fixed-parameter tractability results for {\GoneGtwoDeletion} without any connectivity constraints when some specific $\GG_1$ and $\GG_2$ are given.
Being motivated by their works, we consider {\pgonegtwodel} problem as follows.

\defparproblem{{\pgonegtwodel}}{An undirected graph $G = (V, E)$ and an integer $k$.}{$k$}{Is there a set $S$ with at most $k$ vertices such that $G[S]$ is $p$-edge-connected and every connected component of $G - S$ is either in $\GG_1$ or in $\GG_2$?}

We consider the {\pgonegtwodel} problem when $\GG_1$ and $\GG_2$ both are subclasses of $\eta$-degenerate graphs and satisfy some additional special properties.
These additional special properties are necessary as otherwise we do not have a guarantee that minimal vertex deletion sets can be enumerated in singly-exponential-time. 

\medskip

When we apply our main result Theorem \ref{thm:main-result} to one of the above mentioned problems, there is a common property that it crucially relied on.
The property is that all minimal vertex-deletion sets that must be part of any optimal solution can be enumerated in $2^{\OO(k)}n^{\OO(1)}$-time.
Since a graph of maximum degree $\eta$ is also an $\eta$-degenerate graph, we have the following result as a direct application of our main result.

\begin{restatable}{corollary}{ThmDegree}
\label{corollary:p-edge-bounded-degree}
{\pbdds} admits a $2^{\OO(pk+ \eta)}n^{\OO(1)}$-time algorithm.
\end{restatable}

Our second application is {\ppwoneds} problem.
The graphs of pathwidth at most one are also 2-degenerate.
But it is not straightforward to enumerate all the minimal pathwidth one vertex deletion sets.
So we use some additional characterizations of graphs of pathwidth one and exploit some problem specific structures to prove our next result.

\begin{restatable}{theorem}{ThmPwOne}
\label{thm:pw-1-vd-algo}
{\ppwoneds} admits an algorithm that runs in $2^{\OO(pk)}n^{\OO(1)}$-time.	
\end{restatable}

Note that the algorithm for the above result does not directly invoke the subroutine from Theorem \ref{thm:main-result}.
Instead, it uses some dynamic programming ideas that are closely similar to that of Theorem \ref{thm:main-result} proof but also makes careful local adjustments to take care of some additional constraints.

Our next two applications are {\PECBTDS} and {\ppvc} problems and we have the following two results.

\begin{restatable}{theorem}{ThmTrDepth}
\label{thm:bounded-treedepth}
{\ptdds} admits an algorithm that runs in $2^{2^{2^{\eta}} +\OO(pk + \eta)}n^{2^{2^{\eta}}}$-time.
\end{restatable}

\begin{restatable}{theorem}{ThmBoundedVC}
\label{thm:p-eta-vc-result}
{\ppvc} admits an algorithm that runs in $2^{\OO(pk+\eta)}n^{\OO(1)}$-time.
\end{restatable}

Our last application of {\ExtPCon} is the {\pgonegtwodel} problem under some special conditions.
These specific conditions that we consider for $\GG_1$ and $\GG_2$ are described in Section \ref{sec:scattered-case} and are somewhat technical.
Here we are giving a short description (we refer to Section \ref{sec:prelims} and Section \ref{sec:scattered-case} for more details).

\begin{enumerate}[(i)]
	\item both $\GG_1$ and $\GG_2$ have finite forbidden families $\FF_1$ and $\FF_2$ of induced subgraphs respectively,
	\item there is a fixed constant $\eta$ such that $\GG_1$ and $\GG_2$ are subclasses of $\eta$-degenerate graphs, 
	\item there is a constant $\lambda$ such that $P_{\lambda} \in \FF_1$, and
	\item there are fixed constants $\alpha$ and $\beta$ such that  any $H \in \FF_1$ has at most $\alpha$ vertices and any $H \in \FF_2$ has at most $\beta$ vertices.
\end{enumerate}

Hence, the last result of our paper is the following.

\begin{restatable}{theorem}{ThmScatteredResult}
\label{thm:scattered-result}
Let $\GG_1$ and $\GG_2$ be two distinct hereditary graph classes satisfying conditions {\rm (i), (ii), (iii), (iv)} as described above.
Then, {\pgonegtwodel} admits a deterministic algorithm that runs in $2^{\OO(pk + \eta)}(\alpha + \beta + \lambda)^k n^{\OO(\alpha + \beta + \lambda)}$-time.
\end{restatable}

It remains to be noted that in the absence of these special conditions (i)-(iv) described above, we cannot hope to enumerate all minimal vertex subsets deleting which results every connected component either in $\GG_1$ or in $\GG_2$.
The other cases that are not captured by these conditions are left as open problems for future research directions.

\medskip

\mypara{Organization of our paper:}
We organize the paper as follows.
Initially in Section \ref{sec:prelims}, we introduce the basic notations related to graph theory, parameterized complexity, matroids and representative families.
In addition, we also explain in Section \ref{sec:prelims} how we can test in polynomial-time whether a feasible solution to {\ExtPCon} exists or not.
Then, in Section \ref{sec:p-edge-steiner}, we prove our main result, i.e. Theorem \ref{thm:main-result}.
Then, in Section \ref{sec:applications-p-edge-steiner}, we illustrate the applications of our main result to design singly exponential-time algorithms for {\pbdds}, {\ppwoneds}, {\ptdds}, {\ppvc} and {\pgonegtwodel} under the special conditions.
Finally, in Section \ref{sec:conclusions}, we conclude with some possible future directions of research.

\medskip

\mypara{Related Work:}
Heggernes et al. \cite{HeggernesHMMV15} studied the parameterized comlexity of {\sc $p$-Connected Steiner Subgraph} that is the vertex-connectivity counterpart of our problem.
The authors in their paper have proved that when parameterized by $k$, the above mentioned problem is FPT for $p = 2$ and W[1]-hard when $p = 3$.
Nutov \cite{NutovConnAugment22} has studied a variant of {\sc $p$-Connected Steiner Subgraph} problem in which they have studied {\sc Vertex Connectivity Augmentation} problem.
Given an undirected graph $G$, a $p$-connected subgraph $G[S]$, the {\sc Vertex Connectivity Augmentation} problem asks if at most $k$ additional edges can be added to $G[S]$ to make the subgraph $(p+1)$-connected.
In particular, Nutov \cite{NutovConnAugment22} provided a parameterized algorithm for the above mentioned problem.
Feldmann et al. \cite{Feldmann0L22} have studied parameterized complexity of {\sc Vertex/Edge-Connected Survivable Network Design Problem} where given fixed constant $p$, they want to compute a subgraph that has minimum number of edges and provides $p$-vertex/edge-connectivity between every pair of vertices in the terminals.
Komusiewicz \cite{Komusiewicz16} has conducted a survey on the FPT algorithms to compute cohesive subnetworks but without steiner variants.
This $p$-vertex-connectivity criteria also comes under ``cohesive networks''.

\section{Preliminaries}
\label{sec:prelims}

\subsection{Sets, Numbers and Graph Theory}
\label{subsec:sets-and-graphs}

We use ${\nn}$ to denote the set of all natural numbers.
For $r \in \nn$, we use $[r]$ to denote the set $\{1,\ldots,r\}$.
Given a set $S$ and an integer $k$, we use ${{S}\choose{\leq k}}$ and ${{S}\choose{k}}$ to denote the collection of all subsets of $S$ of size at most $k$ and of size exactly $k$ respectively.
We use standard graph theoretic notations from Diestel's book \cite{Diestel-Book} for all notations of undirected and directed graphs.
For undirected graphs, we use $uv \in E(G)$ to denote that there is an edge between $u$ and $v$.
On the other hand for directed graphs, we are more explicit.
We use $(u, v)$ to represent that the {\em arc} is directed from $u$ to $v$.
For the directed graphs, the directed edges are also called arcs.
We use the term arc and edge interchangeably.
In an undirected graph $G$, we use $deg_G(v)$ to denote the degree of $v$ in $G$. When the graph is clear from the context, we drop this subscript and simply use $deg(v)$.
An undirected graph $G$ is called a {\em degree-$\eta$-graph} if every vertex of $G$ has degree at most $\eta$.
We use $\Delta(G)$ to denote the $\max\{deg_G(v): v \in V(G)\}$, i.e. the maximum degree of any vertex in $G$.
It is clear from the definition that if a graph $G$ is a degree-$\eta$-graph then $\Delta(G) \leq \eta$.
When we consider directed graphs, we have {\em in-degree} and {\em out-degree} for all the vertices.
For a vertex $v$, the {\em in-degree} of $v$ is the number of arcs of the form $(u, v)$ and the {\em out-degree} of $v$ is the number of arcs of the form $(v, u)$.
If $D = (V, A)$ is a directed graph and $F \subseteq A$, then we use $V(F)$ to denote the {\em set of endpoints} of the arcs in $F$. 
{A connected undirected graph $G = (V, E)$ is {\em $p$-edge-connected} if it has at least two vertices and it remains connected after deleting at most $p-1$ edges.}
Due to the Menger's Theorem, a connected graph is said to be $p$-edge-connected if and only if there are $p$ edge-disjoint paths between every pair of vertices.
Given an undirected graph, a set $S \subseteq V(G)$ is said to be a {\em $p$-segment} of $G$ if for every $u, v \in S$, there are $p$ edge-disjoint paths from $u$ to $v$ in $G$.
Note that for a $p$-segment $S$ these $p$ edge-disjoint paths between $u$ and $v$ can use vertices outside $S$.
A $p$-segment is {\em maximal} if it is inclusion-wise maximal.
An undirected graph is said to be an {\em $\eta$-degenerate graph} if every subgraph has a vertex of degree at most $\eta$.
Given an undirected $\eta$-degenerate graph $G = (V, E)$, a sequence of vertices $\rho_G = (v_1,\ldots,v_n)$ is said to be an {\em $\eta$-degeneracy sequence} if for every $2 \leq i \leq n$, $v_i$ has at most $\eta$ neighbors from the vertices $\{v_1,\ldots,v_{i-1}\}$.
For an $\ell \in \nn$, we use $P_{\ell}$ to denote a path containing $\ell$ vertices and $C_{\ell}$ to denote a cycle containing $\ell$ vertices.
It follows from the definition that every degree-$\eta$-graph is an $\eta$-degenerate graph, but the converse does not hold true.
An undirected graph is said to be a {\em caterpillar graph} if every connected component is an induced path with pendant vertices ({\em hairs}) attached to the vertices.
Given two vertices $u, v \in V(G)$, we use ${\rm dist}_G(u, v)$ to denote the {\em distance}, i.e. the number of edges in a shortest path between $u$ and $v$ in $G$.
For two vertex subsets $X, Y \subseteq V(G)$, we define ${\rm dist}_G(X, Y) = \min\limits_{u \in X, v \in Y} \{{\rm dist}_G(u, v)\}$.
Given a directed graph $D = (V, A)$, an {\em out-branching} of $D$ rooted at $v \in V(D)$ is a set of arcs $A' \subseteq A$ such that $v$ has in-degree 0 and every other vertex has in-degree exactly one in $D' = (V, A')$.
Given a class of graphs $\GG$, we say that $\GG$ is {\em polynomial-time recognizable}, if given a graph $G$, there is a polynomial-time algorithm that can correctly check if $G \in \GG$.
A graph class is said to be {\em hereditary} if it is closed under induced subgraphs.

\medskip

We define the following two graph parameters treedepth and pathwidth that we use in our paper.

\begin{definition}(Treedepth)
\label{defn:treedepth}
Given an undirected graph $G = (V, E)$, ${\td}(G)$, i.e. the {\em treedepth} of $G$ is defined as follows.
If $|V(G)| = 1$, then ${\td}(G) = 1$.
If $G$ is connected, then ${\td}(G) = 1 + \min_{u \in V(G)} {\td}(G - \{u\})$.
Finally, if $G_1,\ldots,G_r$ are the connected components of $G$, then ${\td}(G) = \max_{i=1}^r {\td}(G_i)$.
\end{definition}

Informally, a {\em treedepth decomposition} of an undirected graph $G$ can be considered as a rooted forest $Y$ with vertex set $V$ such that for each $uv \in E(G)$, either $u$ is an ancestor of $v$ or $v$ is an ancestor of $u$ in $Y$.
The context of treedepth is also sometimes referred to as {\em elimination tree} of $G$.
It follows from the (recursive) definition above that treedepth of a graph is referred to as the minimum depth of a treedepth decomposition of $G$, where depth is defined as the maximum number of vertices in a root to leaf path.

\begin{definition}[Path Decomposition]
\label{defn:path-decmposition}
A {\em path decomposition} of an undirected graph $G = (V, E)$ is a sequence $(X_1,\ldots,X_r)$ of bags $X_i \subseteq V(G)$ such that
\begin{enumerate}[(i)]
	\item every vertex belongs to at least one bag,
	\item for every edge $uv \in E(G)$, there is $X_i$ such that $u, v \in X_i$, and
	\item for every vertex $v$, the bags containing $v$ forms a contiguous subsequence, i.e. $(X_i,X_{i+1},\ldots,X_j)$. 
\end{enumerate}
The {\em width} of a path decomposition is $\max_{i \in [r]} |X_i| - 1$.
\end{definition}

The {\em pathwidth} of a graph is defined as the smallest number $\eta$ such that there exists a path decomposition of width $\eta$.
Informally, pathwidth is a measure how much a graph is close to a path (or a linear forest).
We use ${\sf pw}(G)$ to denote the pathwidth of $G$.
Another related measure that is widely used {\em treewidth} is how much a graph is close to a tree is defined as follows.

\begin{definition}[Tree Decomposition]
A {\em tree decomposition} of an undirected graph is a pair $(T, {\XX} = \{X_t \mid t \in V(T)\})$ where $T$ is a tree and $\XX$ is a collection of subsets of $V(G)$, called {\em bags} such that
\begin{enumerate}[(i)]
	\item for every edge $uv \in E(G)$, there is a bag $X_t \in \XX$ such that $u, v \in X_t$,
	\item for every vertex $u \in V(G)$, there is a bag $X_t \in \XX$ such that $u \in X_t$, and
	\item for every $u \in V(G)$, the subset of bags $\{t \in T \mid u \in X_t\}$ forms a connected subtree of $T$. 
\end{enumerate}
The {\em width} of a tree decomposition $(T, {\XX})$ is defined as $\max\{|X_t| - 1 : t \in V(T)\}$.
\end{definition}

The {\em treewidth} of a graph $G$ is the smallest number $\eta$ such that there is a tree decomposition of $G$ having width $\eta$.
We use ${\tw}(G)$ to denote the treewidth of $G$.
%
Additionally, we define the following width parameter which is the following.

\begin{definition}[Cutwidth]
\label{defn:cutwidth}
Let $G$ be an undirected graph having $n$ vertices and $\sigma = (v_1,\ldots,v_n)$ be a layout of $V(G)$.
Given $\sigma$ and $i \in \{1,\ldots,n-1\}$, the {\em cut at position $i$} is defined as $\sigma_G(i)$ as the set of edges that have one endpoint in $\{v_1,\ldots,v_i\}$ and the other endpoint in $\{v_{i+1},\ldots,v_n\}$.
The {\em width} of $\sigma$ is denoted by ${\wb}(\sigma) = \max_{i \in [n-1]} \{|\sigma_G(i)|\}$.
The {\em cutwidth} of $G$ is ${\cw}(G) = \min\{{\wb}(\sigma) \mid \sigma$ is a layout of $V(G)\}$.
\end{definition}

For any graph $G$, it follows from the graph theoretic properties that ${\tw}(G) \leq {\cw}(G)$ \cite{ThilikosSB05,BodlaenderFT09,GiannopoulouJLS17}.

\subsection{Parameterized Complexity and W-hardness} 
\label{subsec:FPT-definitions}

A parameterized problem $L$ is a subset of $\Sigma^* \times \nn$ for some finite alphabet $\Sigma$. An instance of a parameterized problem is a pair $(x, k)$ where $x \in \Sigma^*$ is the {\em input} and $k$ is the {\em parameter}.

\begin{definition}[Fixed-Parameter Tractability]
\label{defn:FPT}
A parameterized problem $L \subseteq \Sigma^* \times \nn$ is said to be {\em fixed-parameter tractable} if there exists an algorithm $\AAA$ that given $(x, k) \in \Sigma^* \times \nn$, the algorithm $\AAA$ runs in $f(k)|x|^c$-time for some constant $c$ independent of $n$ and $k$ and correctly decides $L$.
\end{definition}

The algorithm $\AAA$ that runs in $f(k)|x|^{\OO(1)}$-time is called a {\em fixed-parameter algorithm} (or FPT algorithm).
A fixed-parameter algorithm is said to be a {\em singly exponential} FPT algorithm if it runs in $c^k|x|^{\OO(1)}$-time for some fixed constant $c$ independent of $|x|$ and $k$.
There is a hardness theory in parameterized complexity that is associated with the notion of {\em parameterized reduction} and the hierarchy of parameterized complexity classes.
Broadly, the {\em W-hierarchy} (of parameterized complexity classes) is denoted by ${\rm FPT \subseteq W[1] \subseteq W[2] \subseteq \ldots \subseteq XP}$.
Given two distinct parameterized problems $L_1$ and $L_2$, there is a {\em parameterized reduction} from $L_1$ to $L_2$ if given an instance $(x, k)$ of $L_1$, an algorithm $\AAA$ runs in $g(k)|x|^{\OO(1)}$-time and outputs an equivalent instance $(x', k')$ of $L_2$ such that $k' = f(k)$ for some function depending only on $k$.
For more details on parameterized complexity and its associated hardness theory, we refer to the books \cite{CyganFKLMPPS15,Niedermeier06book,DowneyF13Book}.

\subsection{Matroids and Representative Families} 
\label{subsec:matroid-basics}

We use the following definitions and results related to matroid theory to design our algorithms.

\begin{definition}
\label{defn:matroid}
Given a universe $U$ and a subfamily $\cI \subseteq 2^U$,
a set system $\cM = (U, \cI)$ is said to be a {\em matroid} if 
\begin{enumerate}[(i)]
	\item $\emptyset \in \cI$,
	\item if $A \in \cI$, then for all $A' \subseteq A$, $A' \in \cI$ (hereditary property), and
	\item if there exists $A, B \in \cI$ such that $|B| > |A|$, then there is $x \in B \setminus A$ such that $A \cup \{x\} \in \cI$ (exchange property).
\end{enumerate}
The set $U$ is called the {\em ground set} of $\cM$ and a set $A \in \cI$ is called the {\em independent set} of matroid $\cM$.
\end{definition} 

It follows from the definition that all maximal independent sets of a matroid $\cM$ are of the same size.
A maximal independent set of $\cM$ is called a {\em basis}.
Let $U$ be a universe with $n$ elements and $\cI = {{U}\choose{\leq r}}$.
The set system $(U, \cI)$ is called a {\em uniform matroid}.
Let $G = (V, E)$ be an undirected graph and $\cI = \{F \subseteq E(G) \mid G' = (V, F)$ is a forest $\}$.
The set system $(E(G), \cI)$ is called a {\em graphic matroid}.

A matroid ${\cM} = (U, \cI)$ is {\em representable over a field} $\bbF$ if there is a matrix $A$ over $\bbF$ and a bijection $f: U \rightarrow {\rm Col}(A)$ where ${\rm Col}(A)$ is the set of columns of $A$ such that $B \subseteq U$ is an independent set of ${\cM}$ if and only if the set of columns $\{f(b) \mid b \in B\}$ are linearly independent in the matrix $A$.
A matroid representable over a field $\bbF$ is called a {\em linear matroid}.

Given two matroids ${\cM}_1 = (U_1, \cI_1)$ and ${\cM}_2 = (U_2, \cI_2)$, the {\em direct sum} ${\cM} = {\cM}_1 \oplus {\cM}_2$ is the matroid $(U_1 \uplus U_2, \cI)$ such that $I \in \cI$ if and only if $I \cap U_1 \in \cI_1$ and $I \cap U_2 \in \cI_2$.
If ${\cM}_1$ and ${\cM}_2$ are represented by matrices $A_1$ and $A_2$ respectively then $\cM = {\cM}_1 \oplus {\cM}_2$ also admits a matrix representation and can be computed in polynomial-time.
A {\em partition matroid} is a matroid that is formed from directed sum of uniform matroids.

Given a matroid $M$, a {\em truncation} of $M$ to rank $r$ is the matroid $M' = (U, \cI')$ where a set $A \subseteq U$ is independent in $M'$ if and only if $A \in \cI$ and $|A| \leq r$.
Given a matroid $M$ with its representation (in matrix-form), the truncation of $M$ can be computed in polynomial-time.
Let $M = (U, \cI)$ be a matroid and $X, Y \subseteq U$. We say that {\em $X$ extends $Y$} in $M$ if $X \cap Y = \emptyset$ and $X \cup Y \in \cI$.
Moreover, let $\SSS \subseteq 2^U$ be a family.
A subfamily $\widehat{\SSS} \subseteq \SSS$ is a {\em $q$-representative} of $\SSS$ if the following holds: for every set $Y \subseteq U$ with $|Y| \leq q$, there is a set $X \in \SSS$ such that $X$ extends $Y$ if and only if there is a set $\widehat{X} \in \widehat{\SSS}$ such that $\widehat{X}$ extends $Y$.
We use $\widehat{\SSS} \subseteq^q_{\rm rep} \SSS$ to denote that $\widehat{\SSS}$ is a $q$-representative family of $\SSS$.
The following result holds true due to  \cite{FominLPS16,LokshtanovMPS18}.

\begin{proposition}
\label{prop:rep-set-compuatation}
Let $M = (U, \cI)$ be a linear matroid of rank $n$ and $p, q \leq n$ over a field $\bbF$ and let $\SSS = \{S_1,\ldots,S_t\} \subseteq \cI$ each having cardinality $p$.
Then, there exists an algorithm that computes a $q$-representative subfamily $\widehat{\SSS} \subseteq^q_{\rm rep} \SSS$ consisting of at most ${{p+q}\choose{q}}$ sets using $\OO({{p+q}\choose{p}}^2 tp^3 n^2 + t{{p+q}\choose{q}}^{\omega}np) + (n + |U|)^{\OO(1)}$ field operations over $\bbF$.
Here $\omega < 2.37$ is the matrix multiplication exponent.
\end{proposition}

Let $G = (V, E)$ be an undirected graph and $D_G = (V, A_E)$ is defined as follows.
For every $uv \in E(G)$, we add $(u, v)$ and $(v, u)$ into $A_E$ and fix $v_r \in V$.
Since the definition of $D_G$ is based on $G = (V, E)$, we call the pair $(D_G, v_r)$ an {\em equivalent digraph of $G$ with root $v_r$}.
Then, an {\em out-partition matroid with root $v_r$} for $D_G$ is the partition matroid with ground set $A_E$ where arcs are partitioned according to their heads and arcs $(u, v_r)$ are dependent.
Equivalently, what it means is that a set of arcs $F \subseteq A_E$ is {\em independent} in the out-partition matroid with root $v_r$ if and only if $v_r$ has in-degree 0 in $F$ and every other vertex has in-degree at most 1 in $F$.
The {\em graphic matroid on ground set $A_E$} is the graphic matroid for $G$ in which every arc $(u, v)$ represents its underlying undirected edge $uv$ and the antiparallel arcs $(u, v), (v, u)$ represent distinct copies of $uv$.
Then, $\{(u, v), (v, u)\}$ becomes a dependent set.
The following two propositions are proved by Agrawal et al. \cite{AgrawalMPS22} and Einarson et al. \cite{EinarsonGJMW23} respectively.

\begin{proposition}[Agrawal et al. \cite{AgrawalMPS22})]
\label{prop:high-edge-connectivity}
Let $G = (V, E)$ be an undirected graph, $v_r \in V$ and $D_G = (V, A_E)$ as defined above. Then, $G$ is $p$-edge-connected if and only if $D_G$ has $p$ pairwise arc-disjoint out-branchings rooted at $v_r$.
\end{proposition}


\begin{proposition}[Einarson et al. \cite{EinarsonGJMW23}]
\label{prop:out-partition-matroid-connection}
Let $G = (V, E)$ be an undirected graph, $v_r \in V$ and $D_G = (V, A_E)$ as defined above.
Then, $F$ is the arc set of an out-branching rooted at $v_r$ if and only if $|F| = |V(G)| - 1$ and $F$ is independent both in the out-partition matroid for $D_G$ with root $v_r$ and the graphic matroid for $G$ with ground set $A_E$.
\end{proposition}

We refer to Oxley \cite{OxleyMatroidBook} for more details on matroid theory and a survey by Panolan and Saurabh \cite{PanolanS16} for more information on use of matroids in FPT algorithms.

\subsection{Existence of a Highly Connected Steiner Subgraph}
\label{subsec:feasibility-test-1}

Let $G$ be an undirected graph, $X \subseteq V(G)$ and $p$ be an integer.
It is not immediate whether in polynomial-time we can check whether there exists a {\em feasible} solution, that is, a  set $S \supseteq X$ such that $G[S]$ is $p$-edge-connected.
The following lemma illustrates that the above can be achieved in polynomial-time.


\begin{lemma}
\label{lemma:feasibility-check}
Let $G = (V, E)$ be a simple connected undirected graph and $X \subseteq V(G)$.
Then, we can check in polynomial-time if there exists a set $S \supseteq X$ such that $G[S]$ is a $p$-edge-connected subgraph.
\end{lemma}

\begin{proof}
If $G$ is $p$-edge-connected, then we are done as the entire vertex is a trivial solution.
Otherwise, $S_1,\ldots,S_t$ be the collection of all { maximal} $p$-segments of $G$.
They can be computed in polynomial-time \cite{GomoryHu61}.
In particular, it holds true that $S_i \cap S_j = \emptyset$ for all $i \neq j$.
Moreover, for every $i \in [t]$, $|S_i| < |V(G)|$.
If $u \in S_i$ and $v \in S_j$ such that $i \neq j$, then there is an $(u, v)$-cut of size at most $p-1$ in $G$.
Therefore, if $u \in S_i \cap X$ but $v \in S_j \cap X$ for $i \neq j$, then there does not exist a vertex set $S \supseteq X$ such that $G[S]$ is $p$-edge-connected.
So, it must be the case that there is unique $i \in [t]$ such that $X \subseteq S_i$.
Moreover, if there exists $S^* \subseteq V(G)$ such that $G[S^*]$ is $p$-edge-connected, then $S^* \subseteq S_i$ for some $i \leq t$.
Using this observation, we design a recursive algorithm that works as follows.

The base case occurs in two situations.
First situation is when the input graph $G$ is a $p$-edge-connected graph.
Then we output it as a yes-instance and we are done.
Clearly, the entire graph $G$ is a trivial feasible solution.
The second situation occurs when the input graph $G$ has at most $p$ vertices.
As the input is a simple graph, therefore, any $p$-edge-connected subgraph must have at least $p+1$ vertices.
But the input graph $G$ has less than $p+1$ vertices.
So, we output it as a trivial no-instance.

Otherwise, {we are in the case that $G$ has more than $p$ vertices and $G$ is not a $p$-edge-connected subgraph.}
We compute the collection of all {maximal} $p$-segments $S_1,\ldots,S_t$ of $G$ in polynomial-time.
{We have argued in the first paragraph that} there is a unique $i \in [t]$ such that $X \subseteq S_i$.
We recursively call our algorithm with input graph $G[S_i]$ such that $X \subseteq X_i$.
So, this recursive call invokes the recursion at only one $p$-segment $S_i$ such that $X \subseteq S_i$.
%
%
Since $G$ is not $p$-edge-connected, $G$ has at least two $p$-segments each having at least one vertex.
Hence, this algorithm invokes the recursive call for exactly one $p$-segment with strictly smaller number of vertices.
Therefore, this recursive algorithm terminates.

Before invoking the only recursive call, the algorithm checks the base case conditions and computes all maximal $p$-segments that are polynomial-time procedures. 
Hence, this algorithm runs in polynomial-time.
This completes the proof.
\end{proof}

\section{FPT Algorithm for {\ExtPECon}}
\label{sec:p-edge-steiner}

This section is devoted to the proof of the main contribution of our paper.
We first provide a singly exponential algorithm for {\ExtPCon} (we restate below) when the input graph has bounded degeneracy.
We assume that a fixed constant $\eta$ is the degeneracy of $G$.
We restate the problem definition.

\defparproblem{{\ExtPECon}}{A simple undirected graph $G = (V, E)$, $X \subseteq V(G)$ and integers $k, p \in {\nn}$.}{$k + p$}{Is there $S \supseteq X$ of size at most $k$ such that $G[S]$ is $p$-edge-connected?}

Let $(G, X, k)$ be given as an input instance and $\sigma$ be a degeneracy sequence for the vertices of $G - X$ witnessing that $G-X$ is also an $\eta$-degenerate graph. 
Note that one can compute a degeneracy sequence easily in polynomial time by iteratively picking the minimum degree vertex, hence, we assume an ordering $\sigma = (u_1,\ldots,u_n')$ of the vertices of $G - X$ is given along with the input.
Due to Lemma \ref{lemma:feasibility-check}, we can check if there exists a feasible solution $S \supseteq X$ (not necessarily of size at most $k$) such that $G[S]$ is $p$-edge-connected subgraph.
So, we can assume without loss of generality that a feasible solution actually exists.
We first state the following proposition that combines Lemma 4 and Lemma 5 from \cite{EinarsonGJMW23}.

\begin{proposition} (Einarson et al. \cite{EinarsonGJMW23})
\label{prop:matroid-needed-for-us}
{ Let $p$ and $k$ be two integers}, $G = (V, E)$ be an undirected graph, $v_r \in V$, and $D_G = (V, A_E)$ such that $(D_G, v_r)$ is an equivalent digraph with root $v_r$ and $\widehat{M}$ be the direct sum of $2p+1$ matroids as follows.
Matroids $M_1, M_3, \ldots,M_{2p-1}$ are the copies of graphic matroid of $G$ on ground set $A_E$.
Matroids $M_2, M_4, \ldots,M_{2p}$ are the copies of out-partition matroid for $D_G$ with root $v_r$.
Matroid $M_{2p+1}$ is a uniform matroid over $A_E$ with rank $p(k-1)$.
Let $F \subseteq A_E$.
The followings are equivalent.
\begin{enumerate}[(i)]
	\item\label{prop:truncated-prop-1} $F$ is the set of $p$ pairwise arc-disjoint out-branchings rooted at $v_r$ in $D_G[S]$ for some $S \in {{V(G)}\choose{k}}$ and $v_r \in S$.
	\item\label{prop:truncated-prop-2} $|V(F)| = k$, $|F| = p(k-1)$, $v_r \in V(F)$, and there is an independent set $I$ in $\widehat{M}$ such that every arc $a \in F$ occurs in $I$ precisely in its copies in matroids $M_{2i-1}, M_{2i}$ and $M_{2p+1}$ for some $i \in \{1,\ldots,p\}$.
\end{enumerate}

In addition, a linear representation of $\widehat{M}$ can be computed by a deterministic algorithm in polynomial-time.
Subsequently, the truncation of $\widehat{M}$ to a linear matroid ${M}$ of rank $3p(k-1)$ can be computed in deterministic polynomial-time.
Furthermore, every independent set $I$  as described in (\ref{prop:truncated-prop-2}) above is a basis of ${M}$.
\end{proposition} 

Our algorithm for {\ExtPCon} works as follows.
There are two cases.
In the first case, when $X = \emptyset$, we choose an arbitrary vertex $u \in V(G)$ and set $X = \{u\}$.
There are $n = |V(G)|$ possible choices of $X$.
For each such choice we assign $v_r = u \in X$.
Otherwise, it is already the case that $X \neq \emptyset$.
Therefore, we can assume without loss of generality that $X \neq \emptyset$.
In the second case, when $X \neq \emptyset$, we first check whether $G[X]$ is $p$-edge-connected or not.
{If $|X| > k$, then the instance is a trivial no-instance.}
If $G[X]$ is $p$-edge-connected, then we can trivially output yes-instance since $|X| \leq k$.
If $|X| = k$ and $G[X]$ is not $p$-edge-connected, then the instance is a trivial no-instance.
So, we are in the situation that $G[X]$ is not $p$-edge-connected and $|X| < k$.
In the algorithm, we have to check if $X$ can be extended to a $p$-edge-connected subgraph $G[S]$ with at most $k$ vertices.
We work with the equivalent graph $D_G = (V, A_E)$ throughout this section.
We construct a linear matroid $M$ (along with its matrix representation) of rank $3p(k-1)$ in polynomial time that satisfies the properties described in Proposition \ref{prop:matroid-needed-for-us}.
For an independent set $I$ of $M$, we use $V(I)$ to denote the set of endpoints of the arcs that occur in $I$.


\begin{lemma}
\label{lemma:matroid-connection-statement-for-final-proof}
Let $v_r \in X$ and $S \supseteq X$ be such that $|S| = k$.
Then, the following statements are equivalent.
\begin{enumerate}[(i)]
	\item\label{set-existence} $G[S]$ is $p$-edge-connected.
	\item\label{ind-set-properties} there is an independent set $I$  in $M$ such that
	\begin{itemize}
		\item $X \subseteq V(I) = S$, and $|I| = 3p(k-1)$, and
		\item every arc that is in $I$ precisely occurs in three matroids $M_{2i-1}, M_{2i}$, and $M_{2p+1}$ for some $i \in [p]$.
	\end{itemize}
\end{enumerate}
\end{lemma}

\begin{proof}
We fix an arbitrary vertex $v_r \in X$ and $S \supseteq X$ such that $|S| = k$.
First, we give the proof that (\ref{set-existence}) $\Rightarrow$ (\ref{ind-set-properties}).
Let item-(\ref{set-existence}) be satisfied.
{It means that $G[S]$ is $p$-edge-connected.
Then, due to Proposition \ref{prop:high-edge-connectivity}, $D_G[S]$ has a set $F$ of $p$ pairwise arc-disjoint out-branchings rooted at $v_r \in S$ (since $v_r \in X$).
Subsequently, since $|S| = k$ and $D_G[S]$ has a set $F$ of $p$ pairwise arc-disjoint out-branchings rooted at $v_r \in S$, due to the (\ref{prop:truncated-prop-1}) $\Rightarrow$ (\ref{prop:truncated-prop-2}) of Proposition \ref{prop:matroid-needed-for-us}, $v_r \in V(F)$, $|V(F)| = k$, $|F| = p(k-1)$, and there is an independent set $I$ in $M$ such that every arc of $F$ occurs in $I$ precisely in its copies in matroids $M_{2i-1},M_{2i}$, and $M_{2p+1}$ for some $i \in [p]$.
As $F$ is a collection of $p$ pairwise arc-disjoint out-branchings in $D_G[S]$ rooted at $v_r$, it must be that $S = V(F)$.
Since every arc $a \in F$ has precisely three copies in $I$ and $|F| = p(k-1)$, hence $|I| = 3p(k-1)$.
Therefore, the following properties are satisfied. 

\begin{itemize}
	\item $X \subset V(I) = V(F) = S$ and $|V(I)| = k$,
	\item $|I| = 3p(k-1)$, and
	\item every arc that is in $I$ is represented precisely in three matroids $M_{2i-1}, M_{2i}$ and $M_{2p+1}$ for some $i \in [p]$.
\end{itemize}
Hence, we have established the conditions required for item-(\ref{ind-set-properties}).


\medskip

Now, we give the proof that (\ref{ind-set-properties}) $\Rightarrow$ (\ref{set-existence}).
Assume that item-(\ref{ind-set-properties}) is satisfied.
Then, there is an independent set $I$ in the matroid $M$ satisfying the properties of item-(\ref{ind-set-properties}).
Then, $X \subseteq V(I) = S$, $v_r \in V(I)$ and $|S| = k$.
Suppose that $F$ is the set of arcs each of which has precisely three copies in $I$ as described in (\ref{ind-set-properties}).
As $|I| = 3p(k-1)$ and each arc in $F$ appearing in $I$ occurs precisely in three matroids, it follows that $|F| = p(k-1)$ and $V(F) = V(I)$.
In particular, $|V(F)| = k$ and $v_r \in V(F)$.
It ensures that (\ref{prop:truncated-prop-2}) of Proposition \ref{prop:matroid-needed-for-us} is satisfied.
Therefore, due to (\ref{prop:truncated-prop-2}) $\Rightarrow$ (\ref{prop:truncated-prop-1}) of Proposition \ref{prop:matroid-needed-for-us}, $F$ is a set of $p$ pairwise arc-disjoint out-branchings rooted at $v_r \in X$ in $D_G[V(I)]$.
Then, due to Proposition \ref{prop:high-edge-connectivity}, $G[V(I)]$ is $p$-edge-connected and $X \subseteq V(I)$.
Since $|V(I)| = k$, it means that 
there is a set $S = V(I) \supseteq X$ of $k$ vertices such that $G[S]$ is $p$-edge-connected.
Hence, item-(\ref{set-existence}) is satisfied.}
\end{proof}


\medskip

\paragraph{Setting up the Dynamic Programming:}
We will build such a candidate set $I$ via dynamic programming procedure.
Since $X$ is already included in $V(I)$ it allows us to simplify the condition $|V(I)| = |S| = k$ from Lemma \ref{lemma:matroid-connection-statement-for-final-proof} by $|V(I) \setminus X| = k - |X|$.
The purpose of this dynamic programming is to construct a table that keeps track of $|V(I) \setminus X|$ and $|I|$.
Let $n' = |V(G - X)|$ and $\{v_1,\ldots,v_{n'}\}$ be a degeneracy sequence of the vertices of $G - X$ and $\XX = \{A_i \mid A_i = N(v_i) \cap \{v_1,\ldots,v_{i-1}\} \setminus X\}$.
Each entry of the dynamic programming table ${\TT}[(i, j, q, Y), (Z,\ell)]$ will contain a collection of independents sets of $M$ that is a  $(3p(k-1)-q)$-representative family of all the independent sets $I$ of $M$ such that the following conditions are satisfied:
\begin{itemize}
	\item $|V(I) \setminus X| = i$,
	\item $|I| = q$,
	\item the largest index of $V(G) - X$ that occurs in $V(I)$ is $j$,
	\item $Y = A_j \cap V(I)$, and
	\item $Z = A_{\ell} \cap V(I)$ such that $\ell > j$.
\end{itemize} 

Trivially, for a table entry $\TT[(i, j, q, Y), (Z, \ell)]$, if no independent set of $M$ exists that satisfies the above mentioned five conditions, then $\TT[(i, j, q, Y), (Z, \ell)] = \emptyset$.

 Informally, it means that every independent set $|I|$ has size $q$ in $M$, $V(I)$ intersects $A_j$ exactly in $Y$, $V(I)$ spans $i$ vertices from $G - X$, $v_j \in V(I)$, and $V(I)$ has no vertex from $\{v_{j+1},\ldots,v_{n'}\}$.
Furthermore, $V(I)$ intersects $A_{\ell}$ exactly in $Z$ when $\ell > j$.
Observe that for $j = n'$, there is no index $\ell > j$.
Then we denote $\ell = n' + 1$ and $Z = \emptyset$ to keep the DP-states well-defined.
We  prove the following lemma that illustrates how a dynamic programming algorithm can construct all the entries of a table ${\TT}[(i, j, q, Y),(Z, \ell)]$ for $i \leq k - |X|, j \leq n', q \leq 3p(k-1)$, $Y \subseteq A_j$, $Z \subseteq A_{\ell}$ and $\ell > j$.
Observe that there are at most $2^{\eta}n$ possible choices of $Y$ and $2^{\eta}n$ possible choices for $(Z, \ell)$ in the DP table ${\TT}$.
The following lemma illustrates how we compute the DP-table entries.


\begin{lemma}
\label{lemma:generic-table-entry-computation}
Let $M$ be a matroid of rank $r = 3p(k-1)$ as described in Proposition \ref{prop:matroid-needed-for-us} and $v_r \in X$.
Then, the following statements hold true.
\begin{enumerate}
	\item\label{main-table-entry-comput} The entries of the table ${\TT}[(i, j, q, Y),(Z,\ell)]$ for $i \leq k - |X|, j \leq n', j < \ell, q \leq 3p(k-1)$, $Y \subseteq A_j$ and $Z \subseteq A_{\ell}$ can be computed in $2^{\OO(pk + \eta)}n^{\cO(1)}$-time.
	\item\label{main-correctness-arguments} Moreover, there is a set $I \in \TT[(k - |X|, j, 3p(k-1), Y),(Z,\ell)]$ for some $Y \subseteq A_j$ and $Z \subseteq A_{\ell}$ such that $X \subseteq V(I)$ if and only if there exists $S \supset X$ such that $G[S]$ is a $p$-edge-connected subgraph of $G$ with $k$ vertices.
\end{enumerate} 
\end{lemma}

\begin{proof}
We prove the statements in the given order.

\begin{enumerate}
	\item
	We describe a procedure {\sf Construct}(${\TT}[(i, j, q, Y), (Z, \ell)]$) for $i \leq k - |X|, j \leq n'$ and $q \leq 3p(k-1)$ as follows.
	Note that for an arc $a \in A_E$, its copies are present in the out-partition matroids $M_2, M_4, \ldots, M_{2p}$.
	The arc $a$ also has a copy present in the uniform matroid $M_{2p+1}$. 
	The copies of arc $a$ are also present in the graphic matroids $M_1,M_3,\ldots,M_{2p-1}$ on ground set $A_E$.
	Given an arc $a \in A_E$, and $i \in [p]$, we use $F_{a, i}$ to denote the set that contains the copies of $a$ in $M_{2i-1}, M_{2i}$ and $M_{2p+1}$.
	In the first part, for every $q = 3i$ such that $i \in [p(k-1) - 1] \cup \{0\}$, we describe how to construct the table entries $\TT[(0, 0, q, Y),(Z, \ell)]$ as follows.
	
	\begin{enumerate}
		\item For all $1 \leq \ell \leq n'$, we initialize $\TT[(0, 0, 0, \emptyset),(\emptyset, \ell)] = \{ \emptyset \}$.
	
	\item Consider the set of all the arcs in $D_G[X]$. We construct $\TT[(0, 0, q+3, \emptyset),(\emptyset, \ell)]$ from $\TT[(0, 0, q, \emptyset),(\emptyset, \ell)]$ as follows.
		For every $I \in \TT[(0, 0, q, \emptyset),(\emptyset, \ell)]$, for every arc $a \in D_G[X]$, ($1 \leq \ell \leq n'$), and $i \in \{1,\ldots,p\}$, we add $I \cup F_{a, i}$ when $F_{a,i}$ extends $I$.
		
	\item Finally, we invoke Proposition~\ref{prop:rep-set-compuatation} to reduce $\TT[(0, 0, q+3, \emptyset),(\emptyset, \ell)]$ into a $(3p(k-1) - q - 3)$-representative family of size $2^{\OO(pk + \eta)}n^{\OO(1)}$.
	\end{enumerate}
	
	When we consider the table entries $\TT[(i, j, q + 3, Y), (Z, \ell)]$ such that $j = 0$, observe that $Y = Z = \emptyset$.
	The reason is that for any $I \in \TT[(i, j, q + 3, Y), (Z, \ell)]$ with the assumption of $j = 0$ implies that $i = 0$ and no vertex from $G - X$ is part of $V(I)$.
	Since $Y, Z \subseteq V(G) \setminus X$, it must be that $Y = Z = \emptyset$.
	So, we consider only those entries in this first phase.
	
	We analyze the running-time of the above process.
	Given $\TT[(0, 0, q, \emptyset),(\emptyset,\ell)]$, computing $\TT[(0, 0, q+3, \emptyset),(\emptyset, \ell)]$ requires polynomial in the size of $|\TT[(0, 0, q, \emptyset),(\emptyset, \ell)]|$.
	Then, computing a $(3p(k-1) - q - 3)$-representative family requires $2^{\OO(pk)}n^{\OO(1)}$-time.
	
	
	\medskip
	The process of computing table entries for $\TT[(i, j, q, Y), (Z, \ell)]$ for $1 \leq i \leq k - |X|, 1 \leq j \leq n'$, $Y \subseteq A_j$ and $Z \subseteq A_{\ell}$ for $\ell > j$ is more complex and needs more careful approach.
	But the basic principle is similar.
	
	
	
	We process these table slots with a lexicographic ordering of the indices $(i, j, q, \ell)$.
	For every $i \geq 1$, when we compute the collection of independent sets in $\TT[(i, j, q, Y),(Z, \ell)]$, we assume that all the table entries $\TT[(i', j', q', Y'), (Z', \ell')]$ respecting the conditions have been already constructed such that $j' < j$, $j' < \ell'$, $i' < i$, $\ell' < \ell$, and $q' < q$.
	
	We want to include $v_j$ into $V(I)$ such that $|V(I) \setminus X| = i$.
	Suppose that $I' \in \TT[(i', j', q', Y'), (Z', \ell')]$.
	First condition is that $Y' = V(I') \cap A_{j'}$ and $Z' = A_{\ell'} \cap V(I')$.
	We say that $I'$ is {\em extendable} to a set $I \in \TT[((i, j, q, Y),(Z, \ell))]$ if $i = i' + 1$, $\ell' = j$, $Y = Z' \subseteq V(I')$, and $Z \subseteq V(I') \cup \{v_j\}$.
	Moreover, $j$ is the next larger index of the vertex included in $V(I)$ after $v_{j'}$ whose index is $j'$.
	
	\medskip
	\sloppy
	We construct the table entries of $\TT[(i, j, q, Y),(Z, \ell)]$ using the table entries of $\TT[(i', j', q', Y'),(Z', \ell')]$ as follows.
	
	\begin{enumerate}[(i)]
		\item Let $d$ be the number of arcs that are incident to $v_j$ such that either the other endpoints have lower indices or the other endpoint is in $X$.
	    Formally, we consider the arcs that are of one of the two forms.
	    The arcs are either of the form $(v_{j''}, v_j), (v_j, v_{j''}) \in A_E$ such that $j'' < j$, and $v_{j''} \in Y$ and $Y = Z' \subseteq V(I')$.
	    Or the arcs are of the form $(u, v_j), (v_j, u)$ for some $u \in X$.
	     Since $G$ is an $\eta$-degenerate graph, observe that $d \leq 2\eta + 2|X| \leq 2(k + \eta)$.
	     
	     \item We create a set $F$ of arcs as follows.
		For every arc $a$ incident to $v_j$ of the above forms { (as illustrated in the previous item)}, we consider the sets $F_{a, h}$ for every $h \in [p]$.
		We either add $F_{a,h}$ into $F$ for some $h \in [p]$, or do not add $F_{a,h}$ into $F$.
		This ensures that there are $(p+1)^d \leq (p+1)^{2\eta + 2k}$ possible collections of arc-sets.
		
		\item For every nonempty such arc-set $F$ and for every $I' \in \TT[(i', j', q', Y'),(Z', \ell)]$,
		if $F$ extends $I'$ in $M$, $Y = Z' \subseteq V(I')$, $i = i' + 1$, $\ell' = j$, $q = q' + |F|$, and $Z \subseteq V(I') \cup \{v_j\}$, then we add $I' \cup F$ into $\TT[(i, j, q, Y),(Z, \ell)]$.
		From the description, it is clear that $I'$ is extendable to a set $I' \cup F \in \TT[(i, j, q' + |F|, Y),(Z, \ell)]$. and $q = q' + |F|$.
	\end{enumerate}
	Finally, invoke Proposition \ref{prop:rep-set-compuatation} to reduce $\TT[(i, j, q, Y),(Z,\ell)]$ into its $(3p(k-1) - q)$-representative family containing at most $2^{3p(k-1)}$ sets.
	Observe that for every set $I' \in \TT[(i', j', q', Y'),(Z',\ell')]$, there are at most $(p+1)^{2|X| + 2\eta}$ sets $I = I' \cup F$ that are added to the slot $\TT[(i, j, q, Y),(Z,\ell)]$.
	Hence we have that
	$$|\TT[(i, j, q, Y),(Z,\ell)]| \leq (p+1)^{2|X| + 2\eta} 2^{3pk} (p+k+\eta)^{\OO(1)}$$
	
	Since the number of indices is $4^{\eta}n^2$, the above implies that computing all the table entries can be performed in $2^{\OO(pk + \eta)}n^{\OO(1)}$-time.

	\item Now, we prove the second item.
	We first prove the forward direction ($\Rightarrow$).
	Let $I \in \TT[(k - |X|, j, 3p(k-1), Y),(Z, \ell)]$ be an independent set in $M$ for some $j \in [n']$, $Y = A_j \cap V(I)$ and $Z = V(I) \cap A_{\ell}$ such that $v_r \in V(I)$ and $X \subset V(I)$.
	Let $S = V(I)$ and it follows that $|S \setminus X| = k - |X|$.
	Note that $|I| = 3p(k-1)$ and $v_r \in S$.
	When we grow a set $I$, we add a collection of arcs $F_{a, i}$ into it. 
	But, $F_{a, i}$ contains a copy of arc $a$ in $M_{2p+1}$.
	Therefore, construction of $I$ ensures that no two distinct sets $F_{a, i}$ and $F_{a, i'}$ are added together.
	Hence, every arc occurs three times.
	{Suppose that $F$ is the set of arcs such that each arc of $F$ occurs exactly three times in $I$.
	In particular, if $F_{a,i}$ is added into $I$, then the arc $a$ has three copies that appear in $M_{2i-1}, M_{2i}$ and $M_{2p+1}$.
	Then, the properties required by item-(\ref{ind-set-properties}) of Lemma \ref{lemma:matroid-connection-statement-for-final-proof} are satisfied for $I$.
	Due to (\ref{ind-set-properties}) $\Rightarrow$ (\ref{set-existence}) of Lemma \ref{lemma:matroid-connection-statement-for-final-proof}, $G[S]$ is $p$-edge-connected such that $S \supset X$.
	As $V(I) = S$ and $|S| = k$, the correctness follows.}


	\medskip
	Now, we give the backward direction ($\Leftarrow$) of the proof.
	Let $G[S]$ is $p$-edge-connected with $|S| = k$, such that $X \subset S$.
	{It means that item-(\ref{set-existence}) of Lemma \ref{lemma:matroid-connection-statement-for-final-proof} is satisfied.
	Then, due to (\ref{set-existence}) $\Rightarrow$ (\ref{ind-set-properties}) of Lemma \ref{lemma:matroid-connection-statement-for-final-proof}, there is an independent set $I$ in $M$ such that $X \subset V(I) = S$ and $|I| = 3p(k-1)$.
	Moreover, every arc that is in $I$ precisely appears three times.
	Let $F$ be the set of those arcs that are contained in $I$.
	Hence, $|F| = 3p(k-1)$.}
	
	Suppose that $j$ is the largest index from $\{1,2,\ldots,n'\}$ such that $v_j \in V(I)$.
	We complete our proof by justifying that $I$ is a candidate independent set of $M$ for the slot $\TT[(k - |X|, j, 3p(k-1), Y),(Z, \ell)]$ for some $Z = V(I) \cap A_{\ell}$ and $\ell > j$.
	If we prove that there is $Y \subseteq N(v_j) \cap \{v_1,\ldots,v_{j-1}\}$ such that $\TT[(k- |X|, j, 3p(k-1), Y),(Z, \ell)] \neq \emptyset$, then we are done.
	We prove by induction on $i, j, q$ that $\TT[(k- |X|, j, 3p(k-1), Y),(Z, \ell)] \neq \emptyset$.
	
	Note that $I$ is a disjoint union of the sets $F_{a, b}$ for arc $a$ and $b \in [p]$.
	Let us order the arcs of $I$ as per their occurences in lexicographic ordering of $(i, j, q, \ell)$ corresponding to the slots $\TT[(i, j, q, Y),(Z,\ell)]$.
	We first consider the slots $\TT[(0, 0, q, \emptyset),(\emptyset, \ell)]$ for all $1 \leq \ell \leq n'$.
	Suppose that $a_1,\ldots,a_t$ be the arcs of $D_G[X]$ that have been added in this order and we consider the sets $F_{a_i, b_i}$ that are contained in $I$.
	For $s \in [t]$, we define
	$$I_s' = I \setminus (\bigcup\limits_{j=1}^s F_{a_j, b_j})$$
	
	We also define $I_0' = I$.
	Observe that $I_s'$ is the set of arcs that have been added to $I$ after $F_{a_s, b_s}$ in the ordering.
	{Hence by notation, the set $\bigcup\limits_{j=1}^s F_{a_j, b_j}$ extends $I_s'$.}
	
	First, we prove by induction that for each $s \in [t]$, the table slot $\TT[(0, 0, 3s, \emptyset),(\emptyset, \ell)]$ contains an independent set that extends $I_s'$.
	\begin{itemize}
		\item {\bf Base Case:} The base case occurs when $s = 0$. Then, $\TT[(0, 0, 0, \emptyset),(\emptyset, \ell)] = \{\emptyset\}$ and $I_0' = I$.
		Clearly, $\emptyset$ extends $I$ and the base case is proved.
		\item {\bf Induction Hypothesis:} For all $s \leq h < t$, the statement holds true. It means that there is $J_0 \in \TT[(0, 0, 3h, \emptyset),(\emptyset, \ell)]$ that extends $I_h'$.
		\item {\bf Induction Step:}
		Let us consider $s = h+1$ and the table entry $\TT[(0, 0, 3h + 3, \emptyset),(\emptyset, \ell)]$.
		When we considered an independent set from the slot $\TT[((0, 0, 3h, \emptyset),(\emptyset,\ell))]$, the arc $a_{h+1}$ was one of the arcs that was considered for addition into the independent set.
		Since based on the induction hypothesis, $J_0 \in \TT[(0, 0, 3h, \emptyset),(\emptyset, \ell)]$ and $a_{h+1}$ was considered for addition by the algorithm, the arcs $F_{a_{h+1}, b_{h+1}}$ was added into $J_0$.
		Hence, the set $J_1 = J_0 \cup F_{a_{h+1}, b_{h+1}}$ was also added into $\TT[(0, 0, 3h + 3, \emptyset),(\emptyset, \ell)]$.
		As $a_i$'s are distinct arcs and $F_{a_{h+1},b_{h+1}} \subseteq I_h'$, it follows that $F_{a_{h+1}, b_{h+1}}$ extends $J_0$.
		Therefore, the table entry $\TT[(0, 0, 3h+3, \emptyset),(\emptyset, \ell)]$ contains $J_0 \cup F_{a_{h+1},b_{h+1}}$ hence contains $J_1$ before Proposition \ref{prop:rep-set-compuatation} was invoked and $J_1$ extends $I_{h+1}'$.
		Due to the correctness of Proposition \ref{prop:rep-set-compuatation}, it follows that there exists $J_1' \in \TT[(0, 0, 3h+3, \emptyset),(\emptyset, \ell)]$ such that $J_1'$ extends $I_{h+1}'$.
		It ensures us that $\TT[(0, 0, 3h + 3, \emptyset),(\emptyset, \ell)] \neq \emptyset$.
		This ensures that the claim holds true that for all $s \leq t, \TT[(0, 0, 3s, \emptyset),(\emptyset, \ell)] \neq \emptyset$.
	\end{itemize}
	
	Let the vertex set $S \setminus X = \{v_{j_1},\ldots,v_{j_{k- |X|}}\}$ such that $j_1 < \ldots < j_{k - |X|}$.
	We are yet to give a proof that the claim holds true for every $i, j_i, q$ and $Y \subseteq N(v_j) \cap \{v_1,\ldots,v_{j-1}\}$.
	Let ${\sf Arc}(v_{j_i})$ be the set of arcs incident on $v_{j_i}$. define 
	$$L_i = \bigcup_{b \in [p], a \in {\sf  Arc}(v_{j_i})} (F_{a,b} \cap I)$$
	Informally, $L_i$ be the arcs $F_{a, b}$ such that $F_{a, b} \cap I \neq \emptyset$ and $a$ is incident to $v_{j_i}$.
	{Let us also define $P_{i} = \bigcup\limits_{c=i}^{k-|X|} L_c$, i.e. the set of all the copies of arcs incident to the vertices $v_{j_i}, v_{j_{i+1}},\ldots,v_{j_{k - |X|}}$ that are contained in $I$.}
	If we prove by induction that, for every $i \in [k-|X|] \cup \{0\}$, the table slot $\TT[(i, j_i, q_i, Y),(Z, j_{\hat i})]$ for some ${\hat i} > i$ contains a set that extends $P_{i+1}$, then we are done.
	We have already proved the base case $\TT[(0, 0, q_0, \emptyset),(\emptyset, \ell)]$ that occurs in the beginning slots of lexicographical ordering. 
	\begin{itemize}
		\item {\bf Induction Hypothesis:} The claim holds true for all $i \leq s$, i.e. $\TT[(i, j_i, q_i, Y_i),(Z_{i+1}, j_{i+1})]$ contains a set $K_i$ that extends $P_{i+1}$.
		
		\item {\bf Induction Step:}
		Consider the case when $i = s+1$.
		Observe that a nonempty set of arcs incident to $v_{j_{s+1}}$ were added by the algorithm.
		We consider the stage of the algorithm when the entries of the slot $\TT[(s + 1, j_{s+1}, q_{s+1}, Y_{s+1}),(Z_{s+2},j_{s+2})]$ were being constructed.
		A nonempty set of arcs $L_{s+1}$ incident to $v_{j_{s+1}}$ were added to $K_s$ if the the other endpoint of those arcs are in $X$ or $v_{j_{i'}}$ for some $i' \leq s$, $L_{s+1}$ extends $K_s$.
		Subsequently, the set $L_{s+1} \cup K_s$ was added to the table slot $\TT[(s + 1, j_{s+1}, q_{s+1}, Y_{s+1}),(Z_{s+2},j_{s+2})]$.
		Moreover, it also satisfies that $Z_{s+1} = Y_{s+1}$ and $Y_{s+1} \subseteq V(K_s)$.
		
		In particular, observe that $L_{s+1}$ extends $K_s$ and it follows that $K_s \cup L_{s+1} \in \TT[(s + 1, j_{s+1}, q_{s+1}, Y_{s+1}), (Z_{s+2},j_{s+2})]$ before invoking Proposition \ref{prop:rep-set-compuatation}.
		Due to induction hypothesis, $K_s$ extends $P_{s+1}$.
		Furthermore, $P_{s+2} \subseteq P_{s+1}$, $L_{s+1} \subseteq P_{s+1}$, and $L_{s+1} \cap P_{s+2} = \emptyset$, it follows that $K_s \cup L_{s+1}$ extends $P_{s+2}$.
		Once again, due to the correctness of Proposition \ref{prop:rep-set-compuatation}, it follows that there exists $\widetilde K_{s+1}$ that extends $P_{s+2}$ and $|\widetilde K_{s+1}| = |K_s \cup L_{s+1}|$.
		Moreover, this particular set $\widetilde K_{s+1}$ is retained in the slot $\TT[(s + 1, j_{s+1}, q_{s+1}, Y_{s+1}), (Z_{s+2},j_{s+2})]$ by the algorithm in the first item.
	\end{itemize}
	This completes the proof of backward direction that for some $\ell > j_{k-|X|}$, the table slot $\TT[(k - |X|, j_{k - |X|}, 3p(k-1), Y),(Z, \ell)] \neq \emptyset$ implying that there exists some independent set of $M$ that is in $\TT[(k - |X|, j_{k-|X|}, 3p(k-1), Y),(Z,\ell)]$.
	\end{enumerate}
This completes the proof of the lemma.
\end{proof}

Using Lemma \ref{lemma:generic-table-entry-computation}, we are ready to prove our theorem statement of our main result, i.e. Theorem \ref{thm:main-result} (we restate below).

{\ThmMain*}

\begin{proof}
We assume without loss of generality that $X \neq \emptyset$ and $G[X]$ is not $p$-edge-connected.
We use $(G, X, k, p)$ to denote the input instance.
Also, we may assume that we have tested that there is no set $S \supseteq X$ of size at most $k-1$ such that $G[S]$ is $p$-edge-connected. Hence, we may assume that $(G, X, r, p)$ is a no-instance for every $r \leq k-1$.
First, we invoke Lemma \ref{lemma:feasibility-check} to check in polynomial-time whether there exists a subset $S \supseteq X$ such that $G[S]$ is a $p$-edge-connected subgraph or not.
If there is no subset $S \supseteq X$ such that $G[S]$ is a $p$-edge-connected subgraph, then we output that $(G, X, k, p)$ is a no-instance.
Otherwise, we perform the next steps as follows.
Let $\{v_1,\ldots,v_{n'}\}$ be a degeneracy sequence of the vertices of $G - X$ and consider an arbitrary vertex $v_r \in X$.
In the first step, we construct a linear matroid $M$ of rank $3p(k-1)$ and its linear representation that satisfy the properties as described in Proposition \ref{prop:matroid-needed-for-us}.
Our next step is to invoke Lemma \ref{lemma:generic-table-entry-computation} and compute the table entries $\TT[(i, j, q, Y),(Z,\ell)]$ for all $i \in  \{1,\ldots,k-|X|\}$, $j \in [n']$, $q \leq 3p(k-1)$, $Y \subseteq A_j$, $Z \subseteq A_{\ell}$ and $\ell > j$.
It follows from the item (\ref{main-table-entry-comput}) of Lemma \ref{lemma:generic-table-entry-computation} that all the table entries can be computed in $2^{\OO(pk + \eta)}n^{\OO(1)}$-time.
Moreover, it follows from item (\ref{main-correctness-arguments}) of Lemma \ref{lemma:generic-table-entry-computation} that for every $I \in \TT[(k - |X|, j, 3p(k-1), Y),(Z,\ell)]$, the subgraph $G[V(I)]$ is $p$-edge-connected.
Therefore, if the table entry $\TT[(k - |X|, j, 3p(k-1), Y),(Z,\ell)] = \emptyset$, then we return that the input instance is a no-instance.
Otherwise, there is $I \in \TT[(k - |X|, j, 3p(k-1), Y),(Z,\ell)]$ and we output $S = V(I)$ as the solution to the input instance.
This completes the correctness proof of our algorithm.
\end{proof}

As we have proved our main result, we illustrate how our main result provides singly exponential FPT algorithm when $G$ is in several well-known hereditary graph classes.

{\CorMain*}

\begin{proof}
We prove the items in the given order.
For both the items, let $G$ be the input undirected graph.
Moreover, to prove both the items, we just need to ensure that we can obtain a degeneracy sequence of $G$ in polynomial-time.
\begin{enumerate}[(i)]
	\item In the first step, we invoke
	the algorithm by Korhonen \cite{Korhonen21} that computes a tree decomposition
	 $(T, \cX = \{X_t \mid t \in V(T)\})$ of $G$ of width at most $2\eta + 1$.
	 This algorithm takes $2^{\OO(\eta)}n$-time.
	We root the tree $T$ in an arbitrary node $r \in V(T)$.
	We complete a pre-order traversal or breadth first search traversal of $T$ starting from root $r$.
	This traversal gives a sequence $\rho = (t_1,\ldots,t_{|V(T)|})$ of nodes of $T$. 
	Then, for every $i \in [|V(T)|]$, we output $X_{t_i}$, i.e. the vertices in the bag $X_{t_i}$ one by one (without duplicates) and obtain an ordering $\sigma$ of the vertices of the graph.
	Observe that if a vertex $v \in X_{t}$, then $v$ has at most $4\eta + 1$ neighbors, out of which at most $2\eta$ neighbors are in $X_{t}$ itself and at most $2\eta + 1$ neighbors are in $X_{t^*}$ such that $t^*$ is the parent of $t$ when the tree is rooted at $r$.
	Therefore, this gives us a $(4\eta + 1)$-degeneracy sequence of the vertices of $G$.
	Once we have this sequence of vertices, we invoke Theorem \ref{thm:main-result} to get an algorithm for {\ExtPCon} that runs in $2^{\OO(pk + \eta)}n^{\OO(1)}$-time.
	This completes the proof when $\eta$ is the treewidth of the input graph.
	
	\item Suppose that a layout $\sigma = (v_1,\ldots,v_n)$ of width at most $\eta$ is given with the input.
	It follows from the premise that at any position $i \in [n]$, the cut at position $i$ has at most $\eta$ edges.
	We claim that $\sigma$ is an $\eta$-degeneracy ordering of $G$.
	It follows from the definition that cut of position $i$ is the set of edges that has one endpoint in $\{v_1,\ldots,v_i\}$ and the other endpoint in $\{v_{i+1},\ldots,v_n\}$.
	Therefore, $v_{i+1}$ can have at most $\eta$ neighbors in $\{v_1,\ldots,v_i\}$.
	Hence, $\sigma$ itself is an $\eta$-degeneracy sequence of $G$.
	Next, we invoke Theorem \ref{thm:main-result} to obtain an algorithm for {\ExtPCon} that runs in $2^{\OO(pk + \eta)}n^{\OO(1)}$-time.
\end{enumerate}
This completes the proof of the corollary.
\end{proof}

\section{Applications of {\ExtPCon} to Hitting Set Problems on Graphs}
\label{sec:applications-p-edge-steiner}

In this section, we describe some applications of our main result (Theorem \ref{thm:main-result}) in parameterized algorithms.
But before that, we prove the following lemma.
The proof is similar to the proof of Lemma \ref{lemma:feasibility-check}.

\begin{lemma}
\label{lemma:feasibility-applications}
Let $G$ be a simple graph, $p$ be a fixed constant, and $\GG$ be a polynomial-time recognizable hereditary graph class.
Then, there exists a polynomial-time algorithm that can check if there is a set $S \subseteq V(G)$ such that $G - S \in \GG$ and $G[S]$ is $p$-edge-connected.
\end{lemma}

\begin{proof}
If $G$ is $p$-edge-connected, then the entire vertex set is a trivial set whose deletion results in a graph belonging to class $\GG$.
In such a case, we can output that a feasible solution exists.
Similarly, if $G$ has at most $p$ vertices, then there does not exist a vertex subset $S$ that is a $p$-edge-connected subgraph (as at least $p+1$ vertices are needed for a simple graph to become $p$-edge-connected).
In that case, we trivially output it as a no-instance.
Our base cases occur in these two situations.

Otherwise, our algorithm proceeds as follows.
We compute $S_1,\ldots,S_t$ that are the collection of all { maximal} $p$-segments of $G$ in polynomial-time (see \cite{GomoryHu61}).
As any two distinct maximal $p$-segments of $G$ are pairwise vertex-disjoint and for $u \in S_i$ and $v \in S_j$ with $j \neq i$, there is an $(u, v)$-cut of size at most $p-1$, the following must hold true.

If there exists $S^* \subseteq V(G)$ such that $G - S^* \in \GG$ and $G[S^*]$ is $p$-edge-connected, then there is a unique $i \in [t]$ such that $S^* \subseteq S_i$.

If there is $i \in [t]$ such that $G - S_i \notin \GG$, then there does not exist any subset $S'' \subseteq S_i$ such that $G - S'' \in \GG$.
The reason is that $\GG$ is hereditary graph class.
Hence, we discard those $p$-segments $S_i$ such that $G - S_i \notin \GG$.
We consider the remaining $p$-segments $S_i$ one by one for which it satisfies that $G - S_i \in \GG$.
We recursively call this algorithm for each of those specific $p$-segments $S_i$ such that $G - S_i \in \GG$.


If every recursive call outputs a no-instance, then we output that the original instance is also no-instance and the algorithm stops.
If one recursive call outputs a yes-instance, then we output yes-instance and the algorithm stops.  
Observe that, $\sum_{i=1}^t |S_i| = |V(G)| = n$ when $S_1,\ldots,S_t$ are the collection of nonempty $p$-segments.
Hence the recursion tree depth size at most $n$.
And at every level of recursion, the computation is polynomial in $n$.
Hence, the running time of this algorithm is polynomial in $n$.
\end{proof}

The above lemma implies that for any polynomial-time recognizable (hereditary) graph class $\GG$, we can test in polynomial-time if a feasible $p$-edge-connected vertex subset exists whose deletion results in a graph of class $\GG$.
Note that the class of all pathwidth-one graphs is a polynomial-time recognizable graph class.
Moreover, for every fixed constant $\eta$, the class of all degree-$\eta$-graphs, or the class of all $\eta$-treedepth graphs, or the class of all graphs with no $P_{\eta}$ as subgraphs are all polynomial-time recognizable graph classes.
So, we can test the existence of feasible solutions for all our problems in polynomial-time.

\subsection{Singly Exponential FPT Algorithm for {\PEBDDS}}
\label{sec:algo-pebdds}

Now, we explain how Theorem \ref{thm:main-result} implies a singly exponential algorithm for {\PEBDDS} problem.
We formally state the problem below.

\defparproblem{{\PEBDDS} ({\pbdds})}{An undirected graph $G = (V, E)$ and an integer $k$.}{$k$}{Is there $S \subseteq V(G)$ of size at most $k$ such that $\Delta(G - S) \leq \eta$ and $G[S]$ is $p$-edge-connected?}

The following is the first application of our main result (Theorem \ref{thm:main-result}).
As we have a guarantee from Lemma \ref{lemma:feasibility-applications} that we can test if an input graph has a feasible solution to our problems, we assume without loss of generality that the input graph actually has a feasible solution.
We use this assumption for all our subsequent problems.
If the deletion of a vertex subset $S \subseteq V(G)$ results in a degree-$\eta$-graph, then we say that $S$ is a {\em degree-$\eta$-deletion set} of $G$.
We will use this terminology for the rest of this subsection.
The following corollary is the first application of our main result Theorem \ref{thm:main-result}.

{\ThmDegree*}

\begin{proof}
If the input graph $G$ has a vertex $u$ of degree at least $\eta + 1$, then consider the vertex subset $P_u$ as follows.
Let $P_u$ denote the vertex subset that contains $u$ and an arbitrary collection of $\eta + 1$ neighbors of $u$.
At least one vertex from $P_u$ must be part of part of any minimal $\eta$-degree deletion set of $G$ of size at most $k$.

There are two stages in this algorithm. 
The first stage is ``enumeration stage'' and the second stage is ``expand stage''.
{The enumeration stage just enumerates all possible inclusion-wise minimal $\eta$-degree deletion sets of size at most $k$}.
We do this by a bounded search tree procedure as follows.
Find out a vertex $u \in V(G)$ the degree of which is at least $\eta + 1$.
This is a polynomial-time procedure and can be done by iterating every vertex one by one and its adjacency list.
Construct the set $P_u$ by adding $u$ and its arbitrary $\eta + 1$ neighbors.
We branch over the vertex subset $P_u$ containing $\eta + 2$ vertices and enumerate all minimal degree-$\eta$-deletion sets of size at most $k$ of $G$.
This procedure takes $(\eta+2)^{k} n^{\OO(1)}$-time.
Consider any arbitrary minimal $\eta$-degree deletion set $X$ of $G$.
Observe that in the graph $G - X$, every vertex has degree at most $\eta$ and $|X| \leq k$.
Therefore, $G$ is a $(k + \eta)$-degenerate graph and $G - X$ is an $\eta$-degenerate graph.
Hence, any sequence of vertices of $G - X$ is an $\eta$-degeneracy sequence.

For every minimal $\eta$-degree deletion set $X$, our `expand stage' works as follows.
We invoke the algorithm by Theorem \ref{thm:main-result} with input $(G, X, k)$ to check whether $X$ can be extended to a subset that induces a $p$-edge-connected subgraph or not.
Observe that this algorithm by Theorem 1 runs in $2^{\OO(pk + k + \eta)}n^{\OO(1)}$-time since the degeneracy of the graph is at most $k + \eta$.
Hence, the expand stage runs in $2^{\OO(pk + k + \eta)}n^{\OO(1)}$-time.
As $pk + k + \eta$ is $\OO(pk + \eta)$, running time of this expand stage is $2^{\OO(pk + \eta)}n^{\OO(1)}$.
Note that the enumeration stage runs in $(\eta+2)^{k} n^{\OO(1)}$-time, there are $(\eta + 2)^k$ minimal degree-$\eta$-deletion sets, and for every degree-$\eta$-deletion set, the expand stage works in $2^{\OO(pk + \eta)}n^{\OO(1)}$-time.
Therefore, the total running time of the algorithm is $2^{\OO(pk + \eta)} (\eta + 2)^k n^{\OO(1)}$-time.
As asymptotically $\eta + 2 < pk + \eta$, therefore the overall running time of the algorithm is $2^{\OO(pk + \eta)}n^{\OO(1)}$.

If for each of the minimal degree-$\eta$-deletion set $X$, the algorithm by Theorem \ref{thm:main-result} outputs a no-instance, then we say that $(G, k)$ is a no-instance for {\PEBDDS}.
Otherwise, for some minimal $\eta$-degree deletion set $X$, the algorithm by Theorem \ref{thm:main-result} outputs a set $S \supset X$ such that $G[S]$ is $p$-edge-connected.
{Correctness is guaranteed by the fact that any minimal solution of {\pbdds} must contain an inclusion-wise minimal vertex subset $X$ such that $G - X$ is a graph of maximum degree at most $\eta$.}
Since, all the graphs of degree at most $\eta$ is a hereditary graph class, this algorithm works correctly.
This completes the proof.
\end{proof}

\subsection{Singly Exponential FPT Algorithm for {\ppwoneds}}
\label{sec:ppwonevds}

Now, we describe a singly exponential time algorithm for {\PEPWOneVd} problem using Theorem \ref{thm:main-result}. We formally define the problem as follows.

\defparproblem{{\ppwoneds}}{An undirected graph $G = (V, E)$ and an integer $k$.}{$k$}{Is there $S \subseteq V(G)$ of size at most $k$ such that $G[S]$ is $p$-edge-connected and ${\sf pw}(G - S) \leq 1$?}

We introduce a graph $T_2$ that is illustrated in Figure \ref{fig:T2-hairs}. This graph $T_2$ is also known as {\em long-claw} in literature.
We use the following two characterizations of graphs of pathwidth one.

\begin{figure}[t]
\centering
\captionsetup{justification = centering}
	\includegraphics[scale=0.25]{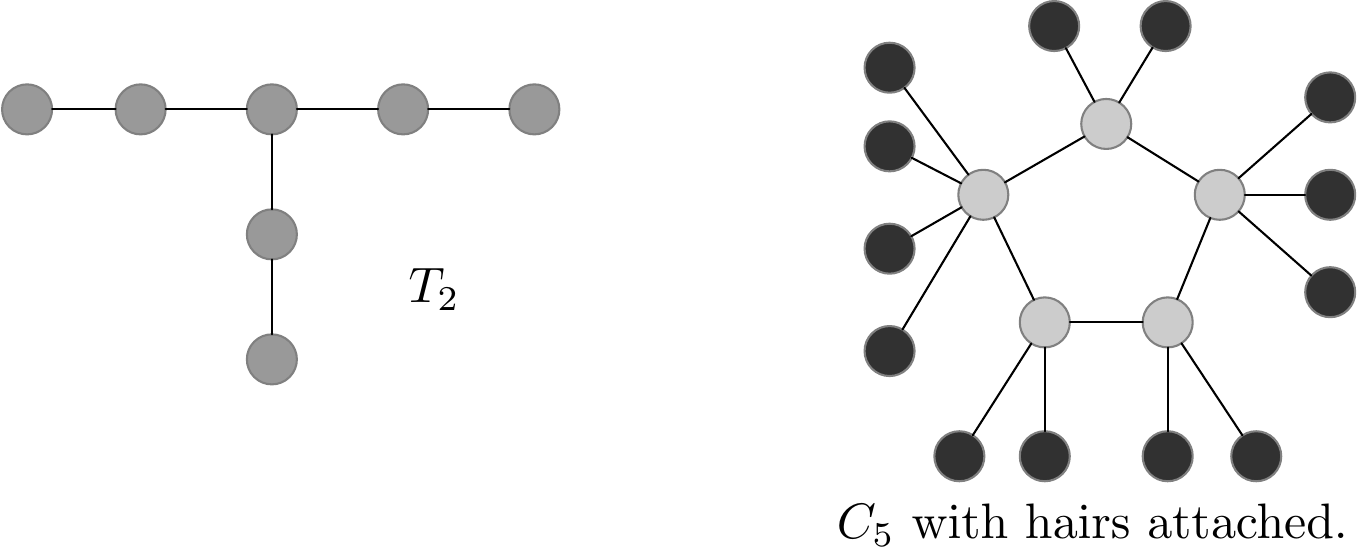}
	\caption{An illustration of $T_2$ (long-claw) and a cycle with hairs attached.}
\label{fig:T2-hairs}
\end{figure}

\begin{proposition}
\label{prop:pw-1-obs-1}
\cite{BurrageEFLMR06,FellowsL89}
A graph $G$ has pathwidth at most one if and only if it does not contain a cycle or $T_2$ as a subgraph.
\end{proposition}

\begin{proposition}
\label{prop:pw-1-obs-2}
\cite{PhilipRV10}
If $G$ is a graph that does not contain any $T_2, C_3, C_4$ as subgraphs, then each connected component of $G$ is either a tree, or a cycle with zero or more pendant vertices (``hairs'') attached to it.
(See Figure \ref{fig:T2-hairs} for an illustration)
\end{proposition}

Observe that the graphs of pathwidth one do not have $T_2, C_3, C_4$ as subgraphs.
We prove the following lemma now.

\begin{lemma}
\label{lemma:pw-1-ordering}
Let $G$ be an undirected graph that does not have any $T_2, C_3, C_4$ as subgraphs.
Then, a 2-degeneracy sequence of $G$ can be constructed in polynomial-time.
Moreover, for every connected component $D$ of $G$, a partition $D = C \uplus P$ can be computed in polynomial-time such that $C$ is an induced path (or cycle) and $P$ is the set of pendant vertices attached to $C$.
\end{lemma}

\begin{proof}
Let the premise of the statement be true.
Since $G$ has no $T_2, C_3, C_4$ as subgraphs, due to Proposition \ref{prop:pw-1-obs-2}, every connected component of $G$ is a tree or a cycle with hairs attached to it.
Since a tree cannot also have $T_2$, every tree component is a caterpillar and it is a 1-degenerate graph.
Similarly, if a component has a cycle then it is a 2-degenerate graph.

Consider a component $D$ that has a cycle.
Let $D$ be a connected component of $G$ and $D = C \uplus P$ such that $C$ is an induced cycle and $P$ is the collection of pendants (hairs) attached to it.
For these components, we can order the vertices participating in the cycle first followed by the hairs attached to it.
For example, let $(v_1,\ldots,v_t)$ be a cycle $C$ and the hairs  attached to $v_i$ are $u_i^1,u_i^2,\ldots,u_i^{\ell_i}$ for all $1 \leq i \leq t$.
Let us consider two ordered tuples $C = (v_1,\ldots,v_t)$ and $P = (\cup_{j=1}^{\ell_i} u_i^j \mid 1 \leq i \leq t)$.
Then, we first put the vertices of $C$ (in the order $v_1,\ldots,v_t$), then we put the vertices of $P$ (in the order $(\cup_{j=1}^{\ell_i} u_i^j \mid 1 \leq i \leq t)$) gives us a 2-degeneracy sequence.
This satisfies the properties of a 2-degeneracy sequence and using standard graph traversal technique they can be computed in polynomial-time.
\end{proof}

It follows from Proposition \ref{prop:pw-1-obs-1} and Proposition \ref{prop:pw-1-obs-2} that any pathwidth one vertex deletion set must intersect all the subgraphs $T_2, C_3, C_4$ of a graph.
But, once we have a set $X$ such that $G - X$ has no $T_2, C_3, C_4$ as subgraphs, then there are some connected components of $G - X$ that can have cycles.
In particular, due to Proposition \ref{prop:pw-1-obs-2}, it holds that if a connected component of $G - X$ has a cycle, then it must be a cycle with some (possibly empty set of) pendant hairs attached to it.
Then, we would need to find $S \supset X$ such that $G[S]$ is $p$-edge-connected and $S$ contains at least one vertex from each of these cycles.
This requires us to design an algorithm that uses the ideas similar to Lemma \ref{lemma:generic-table-entry-computation} but also has to satisfy an additional condition.
We state the following lemma and give a proof for completeness.


\begin{lemma}
\label{lemma:pw-1-extension-part}
Let $G = (V, E)$ be an undirected graph and $X \subseteq V(G)$ such that $G - X$ has no $T_2, C_3, C_4$ as subgraphs and $|X| \leq k$.
Then, there exists an algorithm that runs in $2^{\OO(pk)}n^{\OO(1)}$-time and computes $S \supseteq X$ of size at most $k$ such that $G - S$ has pathwidth at most one and $G[S]$ is a $p$-edge-connected subgraph.
\end{lemma}

\begin{proof}
There are two parts in the proof of this lemma.
The first part of the proof is designing an algorithm that constructs a set of table of entries using dynamic programming over representative families of linear matroid.
The second part is to prove that these constructed table entries are sufficient to correctly find out our desired set (if exists).

We first describe the procedure to construct the table entries using dynamic programming.
In the first stage, we construct a linear matroid $M$, and its linear (matrix)  representation that satisfies the properties of Proposition \ref{prop:matroid-needed-for-us}.
Observe that the input graph $G$ has degeneracy at most $k + 2$.
Proof of this lemma has a lot of similarities with that of Lemma \ref{lemma:generic-table-entry-computation}.
Next step is to invoke Lemma \ref{lemma:pw-1-ordering} to obtain an ordering of the vertices $v_1,\ldots,v_{n'}$ of $G - X$ that gives us a 2-degeneracy sequence.
In particular, if $D_1,\ldots,D_m$ are the components of $G - X$, we order the vertices component-wise.
It means that first we put a 2-degeneracy sequence of a component $D_1$, then we put a 2-degeneracy sequence of $D_2$ and so on.
Moreover, if $D_i = C_i \uplus P_i$ for an induced cycle $C_i$ as in Lemma \ref{lemma:pw-1-ordering}, then we first put the vertices of $C_i$ followed by the vertices of $P_i$.
We assume that $G - X$ has no $T_2, C_3, C_4$ as induced subgraphs and $|X| \leq k$ (as per the premise of this lemma statement).

We now provide a dynamic programming algorithm for the above problem.
We use representative sets to check if there exists $S \supseteq X$ of size at most $k$ such that $S$ intersects $C_i$ in at least one vertex for every $i$.
In particular, we fix an arbitrary vertex $v_r \in X$ and see if there is an independent set $I$ in $M$ such that
\begin{itemize}
	\item $X \subseteq V(I)$,
	\item $|V(I)| = k$,
	\item $|I| = 3p(k-1)$,
	\item every arc that is represented in $I$ precisely appears as its copies in three matroids $M_{2i-1}$, $M_{2i}$ and $M_{2p+1}$ for some $i \in [p]$.
	\item $V(I)$ intersects every $C_i$ in at least one vertex.
\end{itemize}

Observe that the fifth condition is an extra condition compared to what we had in Lemma \ref{lemma:generic-table-entry-computation} proof.
The dynamic programming follows the 2-degeneracy sequence $v_1,\ldots,v_{n'}$ be the ordering of the vertices.
For $i \leq k - |X|, j \leq n', q \leq 3p(k-1)$ and $\ell > j$, we define table entries $\TT[(i, j, q, Y),(Z,\ell)]$ that will contain a collection of independent sets $I$ of $M$ such that the following conditions are satisfied.

\begin{itemize}
	\item $|V(I) \setminus X| = i$,
	\item $|I| = q$,
	\item the largest index in $V(I)$ that occurs in $V(I)$ is $j$
	\item $Y = A_j \cap V(I)$,
	\item $Z = A_{\ell} \cap V(I)$ for any $\ell > j$, and
	\item $V(I)$ intersects every cycle that contains $v_j$ or appears before $v_j$.
\end{itemize}

In the above table entry if $v_j$ is a vertex of some cycle, then $v_j$ can have at most two neighbors.
So, there are 4 possible choices of subsets $Y$.
If $v_j$ is a pendant vertex or the connected component of $G - X$ containing $v_j$ is a caterpillar, then $v_j$ has only one neighbor and there are only two possible choices of $Y$.
Moreover, accordingly, there are at most 4 possible choices of the set $Z$.
Every arc in $I$ occurs exactly in three copies $M_{2i-1}, M_{2i}$ for some $1 \leq i \leq p$ and one copy in $M_{2p+1}$.

In the first step, we construct the table entries $\TT[(0, 0, q, \emptyset),(\emptyset, \ell)]$ for all $1 \leq \ell \leq n'$ and for all $q = 3i$ where $i \in [p(k-1)]$.
The process is the same way as we had done in Lemma \ref{lemma:generic-table-entry-computation}.
\begin{itemize}
		\item For all $1 \leq \ell \leq n'$, we initialize $\TT[(0, 0, 0, \emptyset),(\emptyset, \ell)] = \{ \emptyset \}$.
	
	\item Consider the set of all arcs in $D_G[X]$. We construct $\TT[(0, 0, q+3, \emptyset),(\emptyset, \ell)]$ from $\TT[(0, 0, q, \emptyset),(\emptyset, \ell)]$ as follows.
		For every $I \in \TT[(0, 0, q, \emptyset),(\emptyset, \ell)]$, for every arc $a \in D_G[X]$, ($1 \leq \ell \leq n'$), and $i \in \{1,\ldots,p\}$, we add $I \cup F_{a, i}$ such that $F_{a,i}$ extends $I$.
		
	\item Next, we invoke Proposition \ref{prop:rep-set-compuatation} to reduce $\TT[(0, 0, q+3, \emptyset),(\emptyset, \ell)]$ into a $(3p(k-1) - q - 3)$-representative family of size $2^{\OO(pk + \eta)}n^{\OO(1)}$.
\end{itemize}


By arguments similar to the proof of Lemma \ref{lemma:generic-table-entry-computation}, this procedure takes $2^{\OO(pk)}n^{\OO(1)}$-time.

The process to construct the entries $\TT[(i, j, q, Y), (Z,\ell)]$ is little more complex.
For every $i \geq 1$, we consider a lexicographic ordering of the indices $(i, j, q, \ell)$.
When we compute the table entries of use the table entries of $\TT[(i, j, q, Y),(Z,\ell)]$, we assume that for every $i' < i, j' < j, q' < q, \ell' < \ell$, the table entries $\TT[(i', j', q', Y'),(Z',\ell')]$ have been already computed.
We say that $I' \in \TT[(i', j', q', Y'), (Z', \ell')]$ is {\em extendable} to a set $I \in \TT[(i, j, q, Y),(Z,\ell)]$, if $i = i' + 1$, $Z' = Y \subseteq V(I')$, $\ell' = j$, and $Z \subseteq V(I') \cup \{v_j\}$.
In addition, $j$ is the next larger index vertex included in $V(I)$ after $v_{j'}$.
But in order to capture the condition that our table entries ensure at least one vertex from every cycle intersecting $v_j$ or appearing before $v_j$ is picked into $V(I)$, we have to carefully consider the following cases.

\begin{description}
	\item[Case (i):] The first case is when $v_j$ is in a cycle, i.e. $v_j \in C$.
	In such a case, the definition the table entry of $\TT[(i, j, q, Y),(Z,\ell)]$ ensures that $v_j$ is already added as part of $V(I)$ for any $I \in \TT[(i, j, q, Y),(Z,\ell)]$.
	
	Two subcases also arise.
	\medskip
	
	{\bf Subcase-(i):} The first subcase is $C$ is the first cycle appearing in this $2$-degeneracy sequence of $G - X$.
	We want to include $v_j$ into $V(I)$.
	In this subcase, consider $I' \in \TT[(i', j', q', Y'),(Z',\ell')]$ which satisfies that $Y' = V(I') \cap A_{j'}$ and $Z = V(I') \cap A_{\ell'}$.
	%
	We construct the table entries of $I \in \TT[(i, j, q, Y),(Z,\ell)]$ using the table entries $I' \in \TT[(i', j', q', Y'),(Z',\ell')]$ as follows.
	\begin{itemize}
		\item Let $d$ be the number of arcs that are incident to $v_j$ such that either the other endpoints are is in $X$ or has lower indices.
			Formally, we consider the arcs that are one of the two types.
			The first type of arcs are of the form $(v_{j''}, v_j), (v_j, v_{j''})$ such that $j'' < j$, and $v_{j''} \in V(I')$.
			The second type of arcs are of the form $(u, v_j), (v_j, u)$ for some $u \in X$.
			
		We consider the arcs that are one of the two types.
		First type of arcs are of the form $(v_j, v_{j''}), (v_{j''}, v_j) \in A_E$ such that $v_{j''} \in V(I')$.
		The second type of arcs are of the form $(u, v_j), (v_j, u)$ such that $u \in X$.
		As $d \leq 2(|X| + 2)$, for every arc $a$ of one of these two types incident to $v_j$, we consider the sets $F_{a,h}$ for $1 \leq h \leq p$.
		Either we decide to add $F_{a, h}$ into $F$ or we do not add it at all.
		This ensures that there are $(p+1)^{2|X| + 4}$ choices of $F$.
		
		\item For every such nonempty arc set $F$, and every $I' \in \TT[(i', j', q', Y'),(Z',\ell')]$ if $F$ extends $I'$ in $M$, $Y = Z' \subseteq V(I')$, $\ell' = j$, $q = q' + |F|$,  and $Z \subseteq V(I') \cup \{v_j\}$, then we add $I' \cup F$ into $\TT[(i, j, q, Y), (Z, \ell)]$.
		It is clear from this process that $I'$ is extendable to a set in $\TT[(i, j, q, Y), (Z, \ell)]$ and $q = q' + |F|$. 
	\end{itemize}
	
	\medskip
	
	{\bf Subcase-(ii):} The second subcase is $C$ is not the first cycle in this $2$-degeneracy sequence of $G - X$.
	In that subcase, we look at $I' \in \TT[(i', j', q', Y'),(Z',\ell')]$ with $Y \subseteq V(I')$ that is extendable to a set in $\TT[(i, j, q, Y),(Z,\ell)]$, i.e. $Z' = Y$, $\ell' = j$, and $Z \subseteq V(I') \cup \{v_j\}$.
	It must satisfy a crucial condition.
	Note that $v_{j'}$ is a vertex appearing before $v_j$ in the sequence and there is no vertex $v_{j''}$ such that $j' < j'' < j$, and a nonempty set of arcs is incident to $v_{j''}$ is chosen to be part of $I'$.
	What it means is that the last vertex spanned by the set $I'$ is $v_{j'}$.
	Therefore, it is crucial that there is no cycle $C'$ that is disjoint from $\{v_{j'}, v_j\}$ and uses a set of vertices containing only the indices from $\{j'+1,\ldots,j-1\}$.
	In other words, no other cycle $C'$ should appear after $v_{j'}$ and before $C$.
	Hence, if $C'$ is the previous cycle of $C$ (appearing in the $2$-degeneracy sequence), then $v_{j'}$ must appear either in $C'$ or after $C'$ but before $v_j$.
	We choose only those possible vertices $v_{j'}$ and $I' \in \TT[(i', j', q', Y'),(Z',\ell')]$ with $Y \subseteq V(I')$ that is extendable to a set in $\TT[(i, j, q, Y),(Z,\ell)]$, i.e. $Z' = Y$, $\ell' = j$ and $Z \subseteq V(I') \cup \{v_j\}$.
	
	\begin{itemize}
		\item Let $d$ be the number of arcs incident to $v_j$ that are one of the two types.
		The first type of arcs are of the form $(u, v_j), (v_j, u) \in A_E$ such that $u \in X$.
		The second type of arcs are of the form $(v_{c}, v_j), (v_j, v_c) \in A_E$ such that $c \leq j'$ and $v_c \in V(I')$.
		
		\item For every arc $a$ of the above two types incident to $v_j$, we consider the set $F_{a, h}$ for $1 \leq h \leq p$.
		Either we decide to add $F_{a,h}$ into $F$ or we decide not to add $F_{a, h}$ at all.
		As $d \leq 2(|X|+ 2)$, there are $(p + 1)^{2|X| + 4}$ choices of $F$.
	\end{itemize}
	For every nonempty set of arcs $F$ incident to $v_j$, and every $I' \in \TT[(i', j', q', Y'), (Z', \ell')]$, if $F$ extends $I'$ in $M$, and $q = q' + |F|$, then we add $I' \cup F$ to $\TT[(i, j, q, Y), (Z, \ell)]$.
		
	In both the above two subcases, after we compute $\TT[(i, j, q, F), (Z, \ell)]$, we invoke Proposition \ref{prop:rep-set-compuatation} to shrink the size of $\TT[(i, j, q, Y), (Z, \ell)]$ to a $(3p(k-1) - q)$-representative family containing at most $2^{3p(k-1)}$ sets.
	Since only nonempty arc sets that are incident to $v_j$ are added into $I'$, the definition of $(3p(k-1) - q)$-representative family ensures that even after invoking Proposition \ref{prop:rep-set-compuatation}, the collection of arcs incident on $v_j$ is nonempty.
	
	\item[Case (ii):] The other case is when $v_j$ is a pendant vertex in a component that contains a cycle $C$.	
	
	If $v_j$ is the first pendant vertex in the connected component containing $C$, then we only consider the sets $I' \in  \TT[(i',j',q',Y'),(Z',\ell')]$ such that $v_{j'} \in C$, and $I'$ is extendable to a set in $\TT[(i, j, q, Y),(Z,\ell)]$.
	What is crucial here is that $v_{j'} \in C$.
		
	If $v_j$ is not the first pendant vertex in the corresponding connected component containing $C$, then the previous vertex $v_{j'}$ must be either a pendant vertex from the connected component containing $C$ or from the cycle $C$ inside the same connected component.
	Additionally, we have to check if $I' \in \TT[(i',j',q',Y'),(Z',\ell')]$ is extendable to a set $I \in \TT[(i, j, q, Y),(Z,\ell)]$, i.e. $i' = i+1$, $Z' = Y$, $Z \subseteq V(I') \cup \{v_j\}$, and $\ell' = j$.
	
	Considering the above two situations, we construct $I \in \TT[(i, j, q, F), (Z, \ell)]$ from $I' \in \TT[(i', j', q', F'), (Z', \ell')]$ as follows.
	
	Consider the set of arcs incident to $v_j$ and are one of the two types.
	 The first type is the set of arcs that have other endpoints in $X$, and the second type is the arcs with other endpoint $v_{j''}$ such that $j'' < j$ and $v_{j''} \in V(I')$.
	 For each of these types of arcs $a$ incident to $v_j$, for every $1 \leq h \leq p$, we  construct a nonempty set of arcs $F$ as follows: we either decide to add $F_{a, h}$ to $F$ for some $h \in [p]$ or we do not add $F_{a, h}$ at all.
	 We add the set of arcs $I = I' \cup F$ into $\TT[(i, j, q, Y),(Z,\ell)]$ when $F$ extends $I'$ in $M$ and $q = q' + |F|$.
	 
	 Since the number of arcs incident to $v_j$ is at most $2(|X| + 2)$, therefore there are $(p+1)^{2|X| + 4}$ choices of $F$.
	 	
	Finally, we invoke Proposition \ref{prop:rep-set-compuatation} to reduce the size of $\TT[(i, j, q, Y),(Z,\ell)]$ into a $(3p(k-1) - q)$-representative family having at most $2^{\OO(pk)}$ sets.
	
	\item[Case (iii):] The third case is when $v_j$ appears in a component that is a caterpillar graph.
	This case is also similar to the case (ii).
	We must consider an independent set $I' \in \TT[(i', j', q', Y'), (Z', \ell')]$ such that there is no cycle that is contained in the vertices $\{v_{j'+1},\ldots,v_{j-1}\}$; i.e. no cycles appears after $v_{j'}$ and before $v_j$.
	The computation of table entry process is similar to the previous case, and works as follows.
	
	We consider an arc $a$ incident to $v_j$ that is one of the two forms.
	First type of arc is $(v_j, v_c), (v_c, v_j)$ such that $c \leq j'$ and $v_c \in V(I')$.
	The second type of arc is $(v_j, u), (u, v_j)$ such that $u \in X$.
	For each such arc $a$, consider the set $F_{a, h}$ for $1 \leq h \leq p$.
	Either we decide to add $F_{a, h}$ into $F$ or we decide not to add $F_{a, h}$ at all.
	As $d \leq 2(|X| + 2)$, there are $(p+1)^{2|X| + 4}$ choices of $F$.
	
	For every such nonempty arc set $F$, if $F$ extends $I'$ in $M$,  and $q = q' + |F|$, then we add $F \cup I'$ into $\TT[(i, j, q, Y), (Z, \ell)]$.
	Subsequently, we invoke Proposition \ref{prop:rep-set-compuatation} to reduce the number of sets in the table entry $\TT[(i, j, q, Y), (Z, \ell)]$ to $2^{3p(k-1)}$.
\end{description}


It is clear from the above two cases that the slots have to be processed in the lexicographical ordering of $(i, j, q, \ell)$.
The number of table entries is $2^{\OO(pk)}$ and each table entry computation takes $2^{\OO(pk)}$-time.
Finally, once we compute the table entries for the slot $\TT[(k-|X|, j, 3p(k-1), Y), (Z, \ell)]$ and this slot is nonempty, we check if there is any cycle in $G - X$ that is contained in the subgraph induced by the vertices $\{v_{j+1},\ldots,v_{n'}\}$, i.e.  appearing after $v_j$.
If there is one such cycle, then we output no-instance.
Otherwise, we output the vertex subset $V(I)$ as the solution to our problem.

\medskip

Second part of the proof is to argue the correctness of the above algorithm.
The arguments have some similarities with the proof the second item of Lemma \ref{lemma:generic-table-entry-computation}.
We give a short proof sketch that there exists a set $S \supseteq X$ of size $k$ such that $G[S]$ is $p$-edge-connected if and only if the table entry $\TT[(k - |X|, j, 3p(k-1), Y),(Z,\ell)] \neq \emptyset$.

In the backward direction ($\Leftarrow$), suppose that the table entry $\TT[(k - |X|, j, 3p(k-1), Y), (Z, \ell)] \neq \emptyset$.
Consider $I \in \TT[(k - |X|, j, 3p(k-1), Y), (Z, \ell)] \neq \emptyset$.
The construction of $I$ ensures that if an arc $a$ appears in $I$, then no two distinct sets $F_{a, i}$ and $F_{a, i'}$ was added.
Hence, it follows from Proposition \ref{prop:matroid-needed-for-us} that $D_G[V(I)]$ has $p$ pairwise arc-disjoint out-branchings rooted at $v_r$.
Moreover, due to Proposition \ref{prop:high-edge-connectivity}, $G[V(I)]$ is $p$-edge-connected.
This ensures the $p$-edge-connectivity part.
Additionally, the construction of $I$ also ensures that every cycle $C$ from $G - X$ must have at least one vertex added into $V(I)$.
This property is ensured by the fact that for every cycle $C$ in $G - X$ there exists a vertex $v \in C$ such that a nonempty set of arcs incident to $v$ was added.
Since the arc set is nonempty, the invoke of Proposition \ref{prop:rep-set-compuatation} also ensures that a nonempty set of arcs incident to $v$ was retained.
Hence, $I$ also must have a nonempty set of arcs incident to some vertex in every cycle of $G - X$.
This ensures that $V(I)$ is actually a pathwidth one deletion set of $G$ completing the backward direction of the proof.

The forward direction ($\Rightarrow$) works as follows.
{ Let $S \supseteq X$ be a vertex subset such that $|S| = k$, $G[S]$ is $p$-edge-connected and ${\sf pw}(G - S) \leq 1$.}
Due to Proposition \ref{prop:high-edge-connectivity}, $G[S]$ contains $p$ pairwise arc-disjoint out-branchings rooted at $v_r$.
{
It follows from Proposition \ref{prop:matroid-needed-for-us} that there is an independent set $I$ in $M$ such that $S = V(I)$ containing a set of $p(k-1)$ arcs and $|I| = 3p(k-1)$.}
We complete our proof by justifying that $I$ is a candidate for the table slot $\TT[((k - |X|, j, 3p(k-1), Y), (Z, \ell))]$ for some $Z = V(I) \cap A_{\ell}$ and $\ell > j$.
We use similar arguments as Lemma \ref{lemma:generic-table-entry-computation} to prove that there is $Y \subseteq N(v_j) \cap \{v_1,\ldots,v_{j-1}\}$ such that $\TT[(k - |X|, j, 3p(k-1), Y), (Z, \ell)] \neq \emptyset$ and any candidate for this table entry must contain at least one vertex from every cycle in $G - X$. 

We prove this by induction on $i, j, q$.
Since $I$ is a disjoint union of arc sets $F_{a, b}$ for arc $a$ and $b \in [p]$, we order the arcs as per their occurence in the lexicographic ordering of $\TT[(i, j, q, Y), (Z, \ell)]$.
For the table slots $\TT[(0, 0, q, \emptyset),(\emptyset, \ell)]$, the proof is similar to the second item of Lemma \ref{lemma:generic-table-entry-computation}.

Let the vertex set $S \setminus X = \{v_{j_1},\ldots,v_{j_{k - |X|}}\}$ such that $j_1 < j_2 < \ldots < j_{k - |X|}$.
Suppose that {\sf Arc}($v_{j_i}$) be the set of arcs incident on $v_{j_i}$.
Define

$$L_i = \bigcup\limits_{b \in [p], a \in {\sf Arc}(v_{j_i})} (F_{a, b} \cap I)$$

Equivalently, $L_i$ is the set of arcs $F_{a, b}$ that intersects $I$ and $a$ is incident to $v_{j_i}$.
Same as before, we also define $P_i = \bigcup\limits_{c=i}^{k - |X|} L_c$, i.e. the set of all the copies of arcs incident to the vertices $v_{j_i}, v_{j_{i+1}},\ldots,v_{j_{k - |X|}}$ that are contained in $I$.
We are yet to give a proof by induction that - for every $i \in [k - |X|] \cup \{0\}$, the table slot $\TT[(i, j_i, q_i, Y),(Z, j_{\hat i})]$ for some $\hat i > i$ contains a set that extends $P_{i+1}$ and the vertices spanned by that set contains at least one vertex from every cycle intersecting $\{v_1,\ldots,v_{j_i}\}$.

By induction hypothesis, it already holds true that for all $i \leq s$, the slot $\TT[(i, j_i, q_i, Y_i), (Z_{i+1}, j_{i+1})]$ contains a set $K_i$ that extends $P_{i+1}$ and $V(K_i \cup P_{i+1})$ contains at least one vertex from every cycle that intersects $\{v_1,\ldots,v_{j_i}\}$.
In fact, it is implied by this induction hypothesis that $K_i$ contains a nonempty set of arcs incident to $v_{j_i}$.

At the induction step, we consider $i = s+1$.
Note that our algorithm has added a nonempty set of arcs incident to $v_{j_{s+1}}$.
In this stage, the algorithm has ensured that there is no cycle of $G - X$ should appear after $v_{j_s}$ and before $v_{j_{s+1}}$.
Therefore, every cycle intersecting $\{v_1,\ldots,v_{j_{s+1}}\}$ is covered by the algorithm.
A nonempty set of arcs incident to $v_{j_{s+1}}$ and extends $K_s \in \TT[(s, j_s, q_s, Y_s),(Z_{s+1}, j_{s+1})]$ was added to $K_s$ and the same set was added to the table entry $\TT[(s+1,j_{s+1},q_{s+1}, Y_{s+1}),(Z_{s+2}, j_{s+2})]$ such that $Z_{s+1} = Y_{s+1}$.
Also observe that $L_{s+1}$ extends $K_s$ and $K_s \cup L_{s+1} \in \TT[(s+1, j_{s+1}, q_{s+1}, Y_{s+1}),(Z_{s+2}, j_{s+2})]$ before Proposition \ref{prop:rep-set-compuatation} was invoked.
By induction hypothesis, $K_s$ extends $P_{s+1}$ and $L_{s+1} \subseteq P_{s+1}$.
Hence, it follows that $L_{s+1} \cup K_s$ extends $P_{s+2}$.
Finally, due to the correctness of Proposition \ref{prop:rep-set-compuatation}, it follows that there exists $K_{s+1}$ that extends $P_{s+2}$ and $K_{s+1}$ is also a nonempty set of arcs incident to $j_{s+1}$ and $|K_{s+1}| = |K_s \cup L_{s+1}|$.
Moreover, such a set $K_{s+1}$ is retained by the algorithm.
This completes the proof of the forward direction that the table slot $\TT[(k - |X|, j_{k - |X|}, 3p(k-1), Y),(Z, \ell)] \neq \emptyset$ for some $\ell > j_{k - |X|}$ and an independent set in this table slot must contain at least one vertex from every cycle of $G - X$.
This completes the proof of this lemma.
\end{proof}

Using the above lemma, we provide an $2^{\OO(pk)}n^{\OO(1)}$-time algorithm for {\PEPWOneVd} problem as follows.

{\ThmPwOne*}

\begin{proof}
Let $(G, k)$ be an instance of {\PEPWOneVd} problem.
First we enumerate all minimal vertex subsets $X$ of size at most $k$ such that $G - X$ has no $T_2, C_3, C_4$ as subgraphs.
Since $T_2$ has 7 vertices, $C_3$ has 3 vertices and $C_4$ has 4 vertices, it takes $\OO^*(7^k)$-time to enumerate all such subsets the deletion of which results in a graph that has no $C_3, C_4, T_2$ as subgraphs. Let $X$ be one such set such that $G - X$ has no $C_3, C_4, T_2$ as subgraphs.
Due to Propositions \ref{prop:pw-1-obs-1} and \ref{prop:pw-1-obs-2}, if $D$ is a connected component of $G - X$, then either $D$ is a caterpillar, or a cycle with hairs attached to it.
It follows from Lemma \ref{lemma:pw-1-ordering} that there is a polynomial-time algorithm that gives a $2$-degeneracy sequence $\rho$ of the vertices of $G - X$.
Moreover, if $D$ is a connected component of $G - X$, then $\rho$ provides a partition of $D = C \uplus P$ such that $C$ is a cycle and $P$ is the set of hairs attached to $C$.
In particular, the vertices of $C$ are put first, followed by the vertices of $P$ in $\rho$.
Furthermore, putting the vertices of $C$ first followed by the vertices of $P$ gives a 2-degeneracy ordering of $D$. 
For each such subset $X$, we invoke Lemma \ref{lemma:pw-1-extension-part} to give an algorithm that runs in $2^{\OO(pk)}n^{\OO(1)}$-time and outputs $S$ such that $X \subseteq S$ and $G - S$ has pathwidth at most one.
The correctness of this algorithm also follows from the proof of Lemma \ref{lemma:pw-1-extension-part}.
This completes the proof of this theorem.
\end{proof}

\subsection{Singly Exponential Time FPT Algorithms for {\ptdds} and {\ppvc}}
\label{sec:pedgetdds}

Finally, we describe how we can get $2^{\OO(pk)}n^{\OO(1)}$-time algorithm for {\PECBTDS} and {\pEBoundedVC} problems.
Note that both $p$ and $\eta$ are fixed constants.
 We restate the problem definitions below.

\defparproblem{{\PECBTDS} ({\ptdds})}{An undirected graph $G = (V, E)$ and an integer $k$.}{$k$}{Is there $S \subseteq V(G)$ of size at most $k$ such that $G[S]$ is $p$-edge-connected and ${\td}(G - S) \leq \eta$.}

\defparproblem{{\pEBoundedVC}}{An undirected graph $G = (V, E)$ and an integer $k$.}{$k$}{Is there an $S \subseteq V(G)$ with at most $k$ vertices such that $G - S$ has no $P_{\eta}$ as subgraphs and $G[S]$ is $p$-edge-connected?}

A set $S$ is called an {\em $\eta$-path vertex cover} of $G$ if $G - S$ has no $P_{\eta}$ as subgraph.
{\pEBoundedVC} asks to find an $\eta$-path vertex cover with at most $k$ vertices that induces a $p$-edge-connected subgraph.
The following proposition holds true for graphs of treedepth at most $\eta$.

\begin{proposition}(\cite{DvorakGT12})
\label{prop:treedepth-forbidden}
If a graph $G$ has treedepth at least $\eta + 1$, then it has a connected subgraph $H$ such that ${\td}(H) > \eta$ and $|V(H)| \leq 2^{2^{\eta}}$.
\end{proposition}

The above proposition implies that {\sc $\eta$-Treedepth Deletion Set} problem can be characterized as {\sc ${\cH}$-Hitting Set} problem where ${\cH}$ contains only subgraphs of bounded size.
It follows from the definition that {\ppvc} can be formulated in {\sc $p$-edge-connected $\cH$-Hitting Set} problem.
It means that every minimal $p$-edge-connected $\eta$-treedepth deletion set contains a minimal $\eta$-treedepth deletion set and every minimal every minimal $p$-edge connected $\eta$-path vertex cover contains a minimal $\eta$-path vertex cover.
We prove the following lemma that explains how we can construct a collection of all the minimal such solutions of size at most $k$.

\begin{lemma}
\label{lemma:easy-enumerations}
Given a (connected) undirected graph $G = (V, E)$ and an integer $k$, the collection of all minimal $\eta$-treedepth deletion sets and the collection of all minimal $\eta$-path vertex covers can be computed in $2^{2^{\eta}k}n^{2^{2^{\eta}}}$-time and ${2}^{\OO(\eta)}{\eta}^k n^{\OO(1)}$-time respectively.
\end{lemma}

\begin{proof}
We first prove the statement for enumerating all $\eta$-treedepth deletion sets of $G$.
Due to Proposition \ref{prop:treedepth-forbidden} that the forbidden subgraphs for $\eta$-treedepth graph has at most $2^{2^{\eta}}$ vertices.
Finding an explicit forbidden subgraph takes $\OO(n^{2^{2^{\eta}}})$-time.
Therefore, by bounded search tree method, we can enumerate all the minimal $\eta$-treedepth deletion sets in $2^{2^{\eta}k}n^{2^{2^{\eta}}}$-time.

Similarly, every feasible $\eta$-path vertex cover must intersect all the paths with at most $\eta$ vertices as subgraphs.
It is well known that a path with $\eta$ vertices can be found by a deterministic algorithm in $2^{\OO(\eta)}n^{\OO(1)}$-time (see \cite{ChenKLMRiRoSZ09,GoyalMPZ15} for more details).
Therefore, by bounded search tree method, we can enumerate all the minimal $\eta$-path vertex cover sets in $2^{\OO(\eta)}{\eta}^k n^{\OO(1)}$-time.
\end{proof}

The above lemma implies the next two results two additional applications of our main result.

{\ThmTrDepth*}

\begin{proof}
Let $(G, k)$ be an instance of {\ptdds} problem.
First, we invoke Lemma \ref{lemma:easy-enumerations} to enumerate all possible minimal $\eta$-treedepth deletion sets of $G$ of size at most $k$.
Since the forbidden subgraphs have at most $2^{2^{\eta}}$ vertices, the collection of all minimal $\eta$-treedepth deletion sets of $G$ can be constructed in $2^{2^{\eta}}n^{2^{2^{\eta}}}$-time.
Since ${\td}(G - X) \leq \eta$, this graph $G - X$ is an $\eta$-degenerate graph.
Then, the degeneracy of $G$ is at most $\eta + k$.
Then, for every minimal $\eta$-treedepth deletion set $X$ of $G$ of size at most $k$, we invoke Lemma \ref{lemma:generic-table-entry-computation} and check if there exists a set $S \supseteq X$ such that $G[S]$ is $p$-edge-connected and $|S| \leq k$.
The correctness of this procedure is guaranteed by the second item of Lemma \ref{lemma:generic-table-entry-computation}.
This expand stage takes $2^{\OO(pk + \eta)}n^{\OO(1)}$-time.
Hence, the fixed-parameter algorithm for {\ptdds} runs in $2^{2^{\eta} + \OO(pk+\eta)}n^{2^{2^{\eta}}}$-time.
\end{proof}

Observe that a graph with no $P_{\eta}$ as subgraph has treedepth at most $\eta + 1$. It means that such a graph also has bounded degeneracy. So, we have the following theorem.

{\ThmBoundedVC*}

\begin{proof}
The algorithm has two stages.
First stage is to enumerate all possible minimal $P_{\eta}$-vertex covers of $G$ with at most $k$ vertices.
This can be done by invoking Lemma \ref{lemma:easy-enumerations} and this step takes at most $2^{\OO(\eta)}\eta^k n^{\OO(1)}$-time.
The second stage is expand stage at which the algorithm checks if every minimal $P_{\eta}$-vertex $X$ cover can be extended to a set $S \supseteq X$ such that $|S| \leq k$ and $G[S]$ is $p$-edge-connected.
Note that $G$ is $2\eta + k$ degenerate.
We invoke Theorem \ref{thm:main-result} to perform this step in $2^{\OO(pk +\eta)} n^{\OO(1)}$-time.
Since there are at most $\eta^k n^{\OO(1)}$ distinct minimal $P_{\eta}$-vertex covers, therefore the overall algorithm runs in $2^{\OO(pk + \eta)} \eta^k 2^{\OO(\eta)} n^{\OO(1)}$-time and correctly solves {\pEBoundedVC} problem.
As asymptotically, ${\eta}^k < 2^{\OO(pk + \eta)}$, therefore, the overall running time of the algorithm is $2^{\OO(pk +\eta)} n^{\OO(1)}$-time.
\end{proof}


\subsection{Singly Exponential FPT Algorithm for Deletion to Pairs of Graph Classes}
\label{sec:scattered-case}

In this section, we illustrate how Theorem \ref{thm:main-result} can be used to obtain singly exponential-time algorithms for {\pgonegtwodel} problem when $\GG_1$ and $\GG_2$ satisfy some special properties.
We restate the problem below.

\defparproblem{{\pgonegtwodel}}{An undirected graph $G = (V, E)$ and an integer $k$.}{$k$}{Is there a set $S$ with at most $k$ vertices such that every connected component of $G - S$ is either in $\GG_1$ or in $\GG_2$ and $G[S]$ is $p$-edge-connected?}

\medskip

Given a hereditary class of graphs $\GG$, we say that $\FF$ is the {\em forbidden family} of $\GG$ if a graph $G \in \GG$ does not contain any $H \in {\FF}$ as induced subgraphs.
If for all $H \in \FF$, the number of vertices in $H$, i.e. $|H|$ is a fixed constant independent of the number of vertices in the input graph, then we say that $\FF$ is a {\em finite forbidden family}.
An undirected graph $G$ is said to be a {\em $P_{\eta}$-free} graph if it has no $P_{\eta}$ as subgraph.
Let $\GG$ be the class of all $P_{\eta}$-free graphs.
Then, the size of any obstruction is at most $\eta$ implying that the number of graphs in $\FF$ should be $f(\eta)$.
So, the class of all $P_{\eta}$-free graphs have finite forbidden family.
Similarly, both the classes of all $\eta$-treedepth graphs and the class of all degree-$\eta$-graphs also have finite forbidden families each.
But observe that the forbidden family for graphs having pathwidth at most one is not finite.
The reason behind this is that a pathwidth one graph cannot have any cycle as subgraph and the length of the cycle is dependent on the number of vertices of the input graph.

\medskip

In this section, we provide fixed-parameter tractability results for {\pgonegtwodel} when ${\GG}_1$ and ${\GG}_2$ satisfy the following conditions.
\begin{enumerate}[(i)]
	\item both $\GG_1$ and $\GG_2$ have finite forbidden families $\FF_1$ and $\FF_2$ respectively,
	\item there is a fixed constant $\lambda$ such that $P_{\lambda} \in \FF_1$,
	\item there are fixed constants $\alpha$ and $\beta$ such that the number of vertices in any $H \in \FF_1$ is at most $\alpha$ and the number of vertices in any $H \in \FF_2$ is at most $\beta$, and
	\item there is a fixed constant $\eta$ such that $\GG_1$ and $\GG_2$ are subclasses of $\eta$-degenerate graphs, 
\end{enumerate}

When $\GG_1$ and $\GG_2$ with forbidden families $\FF_1$ and $\FF_2$ respectively satisfy the above four properties, it follows that the number of forbidden pairs $\FF_1 \times \FF_2$ is also finite.
Given a graph $G$, we look at all the pairs of obstructions in $\FF_1 \times \FF_2$.
We say that $(H_1, H_2) \in \FF_1 \times \FF_2$ is a {\em closest} forbidden pair in $G$ if $G$ contains vertex subsets $J_1$ and $J_2$ isomorphic to $H_1$ and $H_2$ respectively such that for any other pair of vertex subsets $J_1'$ isomorphic to forbidden subgraph $H_1' \in \FF_1$, and $J_2'$ isomorphic to forbidden subgraph $H_2' \in \FF_2$, it holds that ${\rm dist}(J_1, J_2) \leq {\rm dist}(J_1', J_2')$.
We use the following property proved by Jacob et al. \cite{JacobMR23}.

\begin{proposition}\cite{JacobMR23}
\label{prop:scattered-property-with-path}
Let $(H_1, H_2) \in \FF_1 \times \FF_2$ be a closest forbidden pair in $G$ such that $(J_1, J_2)$ be the vertex subsets corresponding to the forbidden pair $(H_1, H_2)$ and $P$ be a shortest path between $J_1$ and $J_2$ in $G$.
Moreover, there is a fixed constant $\lambda$ such that $P_{\lambda} \in \FF_1$.
Then, the number of vertices in $P$ is at most $\lambda$.
\end{proposition}

If $X \subseteq V(G)$ such that every connected component of $G - X$ is either in $\GG_1$ or in $\GG_2$, we say that $X$ is a {\em $(\GG_1, \GG_2)$-deletion set} of $G$.
We use the above proposition to prove the following lemma.

\begin{lemma}
\label{lemma:enum-min-gonetwo-del}
Let $\GG_1$ and $\GG_2$ be hereditary graph classes satisfying the properties (i), (ii), (iii), and (iv) as stated above.
Given an undirected graph $G = (V, E)$, the collection of minimal $(\GG_1, \GG_2)$-deletion sets of $G$ can be enumerated in $(\alpha + \beta + \lambda)^k n^{\OO(\alpha + \beta + \lambda)}$-time.
\end{lemma}

\begin{proof}
Suppose that the premise of the statement holds true and let $\FF_1$ and $\FF_2$ be the forbidden families for the graph classes $\GG_1$ and $\GG_2$ respectively.
Let $J_1^*$ be a vertex subset isomorphic to $H_1$ and $J_2^*$ be a vertex subset isomorphic to $H_2$ such that $H_1 \in \FF_1$ and $H_2 \in \FF_2$.
Moreover, let $(J_1^*, J_2^*)$ be a closest forbidden subgraph pair in $G$ such that $P$ be a shortest path between $J_1^*$ and $J_2^*$ in $G$.
It follows from Proposition \ref{prop:scattered-property-with-path} that $P$ has at most $\lambda$ vertices.
Therefore, any mimimal $(\GG_1, \GG_2)$-deletion set must contain at least one vertex from $J_1^* \cup J_2^* \cup P$.
Such a closest obstruction pair $(J_1^*, J_2^*)$ can be identified in $n^{\OO(\alpha + \beta)}$-time.
We branch over every vertex $u \in J_1^* \cup J_2^* \cup P$ one by one and compute all the minimal $(\GG_1, \GG_2)$-deletion sets of size at most $k$.
Since, the number of vertices in $P$ is at most $\lambda$, the bounded search tree has $\alpha + \beta + \lambda$ branches and depth $k$.
Moreover every node, the algorithm takes $n^{\OO(1)}$ additional time since $\alpha, \beta, \lambda$ are all fixed constants.
Hence, the set of all these minimal $(\GG_1, \GG_2)$-deletion sets can be enumerated in $(\alpha + \beta + \lambda)^k n^{\OO(1)}$-time.
\end{proof}

Using the above lemma, we are ready to prove the following result.

{\ThmScatteredResult*}

\begin{proof}
Let $(G, k)$ be an instance to the {\pgonegtwodel} problem such that $\GG_1$ and $\GG_2$ are two hereditary graph classes satisfying the properties (i), (ii), (iii), and (iv).
Our first step is to invoke Lemma \ref{lemma:enum-min-gonetwo-del} to enumerate all minimal $(\GG_1, \GG_2)$-deletion sets of $G$.
If $X$ is a minimal $(\GG_1, \GG_2)$-deletion set, then the correctness of Lemma \ref{lemma:enum-min-gonetwo-del} ensures that every connected component of $G - X$ is either in $\GG_1$ or in $\GG_2$.
Due to condition (iv), every connected component of $G - X$ is an $\eta$-degenerate graph.
In the second step, for every minimal $(\GG_1, \GG_2)$-deletion set $X$ of $G$ with at most $k$ vertices, we invoke our main result, i.e. Theorem \ref{thm:main-result} to compute a vertex subset $S \supset X$ such that $G[S]$ is $p$-edge-connected and $|S| \leq k$.
This completes the proof.
\end{proof}

The following application is a natural candidate example for {\pgonegtwodel} problem.
Consider $\GG_1$ as the class of all graphs of degree at most $\alpha$ and $\GG_2$ as the class of all graphs that has no $P_{\beta}$ as subgraphs.
Observe that both $\GG_1$ and $\GG_2$ satisfy the properties (i), (ii), (iii), and (iv) described above.
Therefore, we have the following corollary as an application to {\pgonegtwodel} problem.

\begin{corollary}
\label{corollary:scattered-first-application}
Let $G$ be an undirected graph with $\alpha$ and $\beta$ being fixed constant.
Then, there is an algorithm that runs in $2^{\OO(pk + \alpha + \beta)}n^{\OO(\alpha + \beta)}$-time and either computes a set $S \subseteq V(G)$ of at most $k$ vertices such that $G[S]$ is $p$-edge-connected and every connected component of $G - S$ is either a degree-$\alpha$-graph or has no path with $\beta$ vertices as subgraph or outputs that no such set exists.
\end{corollary}

\begin{proof}
Let $G$ be an undirected graph and $k$ being a positive integer.
Moreover, suppose that $\alpha$ and $\beta$ are fixed constants and let $\eta = \max\{\alpha,\beta\}$.
Let $\GG_1$ denotes the class of all graphs of degree at most $\alpha$ and $\GG_2$ denotes the class of all graphs that have no path having $\beta$ vertices.
Let $\FF_1$ denotes the forbidden families for $\GG_1$ and $\FF_2$ denotes the forbidden families for $\GG_2$.
Since both $\FF_1$ and $\FF_2$ are finite forbidden families of size at most $\alpha$ and $\beta$ respectively, $\FF_2$ contains $P_{\beta}$, and both $\GG_1$ and $\GG_2$ are subclasses of $\eta$-degenerate graphs, it follows that $\GG_1$ and $\GG_2$ satisfy the conditions (i)-(iv).
Hence, we invoke the algorithm by Theorem \ref{thm:scattered-result} that solves {\pgonegtwodel} when $\GG_1$ is the class of all degree-$\alpha$-graphs and $\GG_2$ is the class of all graphs having no path of $\beta$ vertices.
 This algorithm runs in $2^{\OO(pk + \alpha + \beta)}n^{\OO(\alpha + \beta)}$-time and either it compute a set $S \subseteq V(G)$ such that $G[S]$ is $p$-edge-connected and every connected component of $G - S$ is either a graph of degree at most $\alpha$ or a graph that has no path with $\beta$ vertices.
This completes the proof.
\end{proof}

\section{Conclusions and Future Work}
\label{sec:conclusions}

There are several possible directions of future work.
Our main result proves that {\ExtPCon} is FPT when the removal of terminals results in a bounded degenerate graph.
But this does not cover several other well-known graph classes.
So, the following open problems are obvious.

\begin{open problem}
\label{opp1}	
{\rm Is {\ExtPCon} fixed-parameter tractable when $G$ is an arbitrary graph and $p$ is a fixed constant?}
\end{open problem}

Since $\eta$-degenerate graphs do not capture the class of all $\eta$-cliquewidth graphs, the following open problem arises.

\begin{open problem}
\label{opp2}
{\rm Is {\ExtPCon} fixed-parameter tractable when $G$ is a bounded cliquewidth graph?}	
\end{open problem}

A positive answer to Open Problem \ref{opp2} will imply an FPT algorithm for {\sc $p$-Edge-Connected Cluster Vertex Deletion} and {\sc $p$-Edge-Connected Cograph Vertex Deletion}.

If $p = 2$, then finding $p$-vertex/edge-connected steiner subgraph admits $k^{\OO(k)}n^{\OO(1)}$-time algorithm \cite{HeggernesHMMV15,Feldmann0L22}.
Einarson et al. \cite{EinarsonGJMW23} proved that $p$-Vertex-Connected Steiner Subgraph is FPT when the set of terminals forms a vertex cover of the input graph.
But their algorithm is super-exponential with running time $2^{\OO(k^2)}n^{\OO(1)}$.
So, the following is a natural open problem.

\begin{open problem}
\label{opp3}
{\rm When the set of terminals is a vertex cover of the input graph, can we compute a $p$-vertex-connected steiner subgraph in singly exponential-time?}	
\end{open problem}

We want to emphasize that the status of Open Problem \ref{opp3} is unknown even for $p = 2$.
On the perspective of applications of Theorem \ref{thm:main-result}, we have been successful in designing singly exponential-time FPT algorithms for {\pbdds}, {\ptdds}, {\ppwoneds} and {\ppvc}.
But, our results do not capture several other graph classes.
For instance, the above algorithm crucially relies that all minimal vertex deletion sets without connectivity requirements can be enumerated in $2^{\OO(k)}n^{\OO(1)}$-time and that a bounded degeneracy sequence can be computed in polynomial-time and this does not hold true for {\sc Feedback Vertex Set}.
So, the last and final open problem is the following.

\begin{open problem}
\label{opp4}
{\rm Given an undirected graph $G = (V, E)$ and a fixed-constant $p$, can we compute a $p$-edge-connected feedback vertex set of $G$ with at most $k$ vertices in $2^{\OO(pk)}n^{\OO(1)}$-time (or even in $f(pk)n^{\OO(1)}$-time)?}	
\end{open problem}



\end{document}